\documentclass[aps,pra,twocolumn,superscriptaddress,10pt,nofootinbib]{revtex4-2}
\usepackage{multirow,amsthm,amssymb,amsbsy,amsmath,epsfig,epstopdf,bm,float}
\DeclareMathOperator{\Tr}{Tr}
\usepackage{enumitem}
\usepackage{graphicx}
\usepackage{booktabs}
\usepackage{amsmath}
\usepackage{verbatim}
\usepackage{dcolumn}
\usepackage{color}
\usepackage[dvipsnames]{xcolor}
\usepackage{mathtools}
\usepackage[normalem]{ulem}
\usepackage{soul}
\usepackage{hyperref}
\usepackage{braket}
\usepackage[export]{adjustbox}
\DeclarePairedDelimiter\abs{\lvert}{\rvert}%

\newcommand{\R}{\mathbb{R}}
\newcommand{\C}{\mathbb{C}}
\newcommand{\N}{\mathbb{N}}

\newcommand\restr[2]{{
  \left.\kern-\nulldelimiterspace 
  #1 
  \vphantom{\big|} 
  \right|_{#2} 
  }}
  
\theoremstyle{definition}
\newtheorem{definition}{Definition}

\theoremstyle{remark}

\theoremstyle{prop}
\newtheorem{prop}{Proposition}

\newtheorem{thm}{Theorem}

\begin{document}

\title{Strong symmetries in collision models and physical dilations of covariant quantum maps
}

\author{Marco Cattaneo}
\email{marco.cattaneo@helsinki.fi}
\affiliation{QTF Centre of Excellence,  
Department of Physics, University of Helsinki, P.O. Box 43, FI-00014 Helsinki, Finland}

\date{\today}

\begin{abstract}
Quantum maps are fundamental to quantum information theory and open quantum systems. \textit{Covariant} or \textit{weakly symmetric} quantum maps, in particular, play a key role in defining quantum evolutions that respect thermodynamics, establish free operations in resource theories, and are consistent with transformations of quantum reference frames. To implement quantum maps in the lab, one typically engineers a physical dilation, which corresponds to a unitary evolution entangling the system with an \textit{environment}. This work systematically explores how weak symmetries of quantum maps manifest in their dilations. We demonstrate that for various classes of physical dilations, including Hamiltonian-driven dilations and short-time collision models that simulate Markovian open quantum dynamics, weak symmetries always lead to \textit{strong} symmetries in the dilated evolution, resulting in conserved quantities in the system-environment space. We also characterize the subspace where these symmetries arise using Krylov subspaces. Moreover, we show that some different types of physical dilations have no constraints on the dilated evolution, requiring no strong symmetry.  Finally, we complement our findings with a variety of illustrative and pedagogical examples. Our results provide essential guidelines for constructing physical dilations of quantum maps, offering a comprehensive understanding of how symmetries shape their implementations in a laboratory or on a quantum computer.
\end{abstract}
\maketitle

\section{Introduction}
\label{sec:intro}
Quantum maps are the building blocks of the theories of quantum information and open quantum systems \cite{breuer2002theory,rivas2012open}, providing the most general framework for describing how quantum states evolve over time or transform under physical operations \cite{nielsenchuang,Holevo2012}. While quantum maps are often defined from abstract principles, the state transformations and evolutions we can control in the laboratory are typically unitary, at least in ideal conditions. Examples include physical dynamics governed by the Schrödinger equation, photons traveling through the optical table in quantum optics labs \cite{mandel1995optical}, or quantum algorithms executed on a quantum computer, which are constructed through sequences of unitary gates \cite{nielsenchuang}. Therefore, the implementations of general non-unitary quantum maps in the lab usually involves \textit{physical dilations}, which means that the quantum system on which the map acts must be coupled with another subsystem, which we term \textit{environment}, in such a way that their dynamics is driven by a unitary evolution in the total system-environment space that we can easily realize experimentally. The theory of dilations of quantum maps is based on a celebrated mathematical result known as \textit{Stinespring's theorem} \cite{Stinespring1955,Paulsen2003}, which has been successfully applied to quantum systems. As a result, implementing physical dilations of any quantum map has become a well-established procedure \cite{nielsenchuang,Caruso2006}.

The quantum simulation of open systems is probably the simplest example of physical dilations that are of utmost importance for current quantum technologies \cite{Lloyd2001,Koniorczyk2006,Muller2011,Rybar2012,Dive2015,Cleve2017,Hu2020,DelRe2020,Kamakari2021,Schlimgen2021,Cattaneo2021,Pocrnic2023,DiBartolomeo2023,Ding2024,Delgado-Granados2024}. Dilation-based simulations of open quantum systems have been successfully realized on both analog experimental platforms \cite{Barreiro2011,Schindler2013,Xin2017,Han2021} and digital quantum computers \cite{Wei2018,Garcia-Perez2020,Burger2022,Erbanni2023,Cattaneo2023,Kamakari2021,Hu2020,David2024,DelRe2020}. A particularly relevant dilation-based method for simulating open quantum systems is the use of collision models \cite{Campbell2021a, Ciccarello2021, Cattaneo2022d,Cusumano2022}, which involve repeated interactions between the system and independent particles from the environment. It can be shown that they are a universal procedure to simulate the Markovian dynamics \cite{breuer2002theory,rivas2012open,Alicki2007} of any structured open quantum system by using a finite-dimensional environment \cite{Cattaneo2021,Pocrnic2023}. In recent years, collision models have gained significant importance in quantum information and open systems due to their broad applications in quantum thermodynamics \cite{Barra2015, Strasberg2017, DeChiara2018a}, the study of non-Markovianity \cite{Ciccarello2013a, Vacchini2014, Kretschmer2016a}, and various other fields, as detailed for instance in the references of~\cite{Ciccarello2021, Cattaneo2021}. 

\begin{figure*}
    \centering
    \includegraphics[scale=0.55]{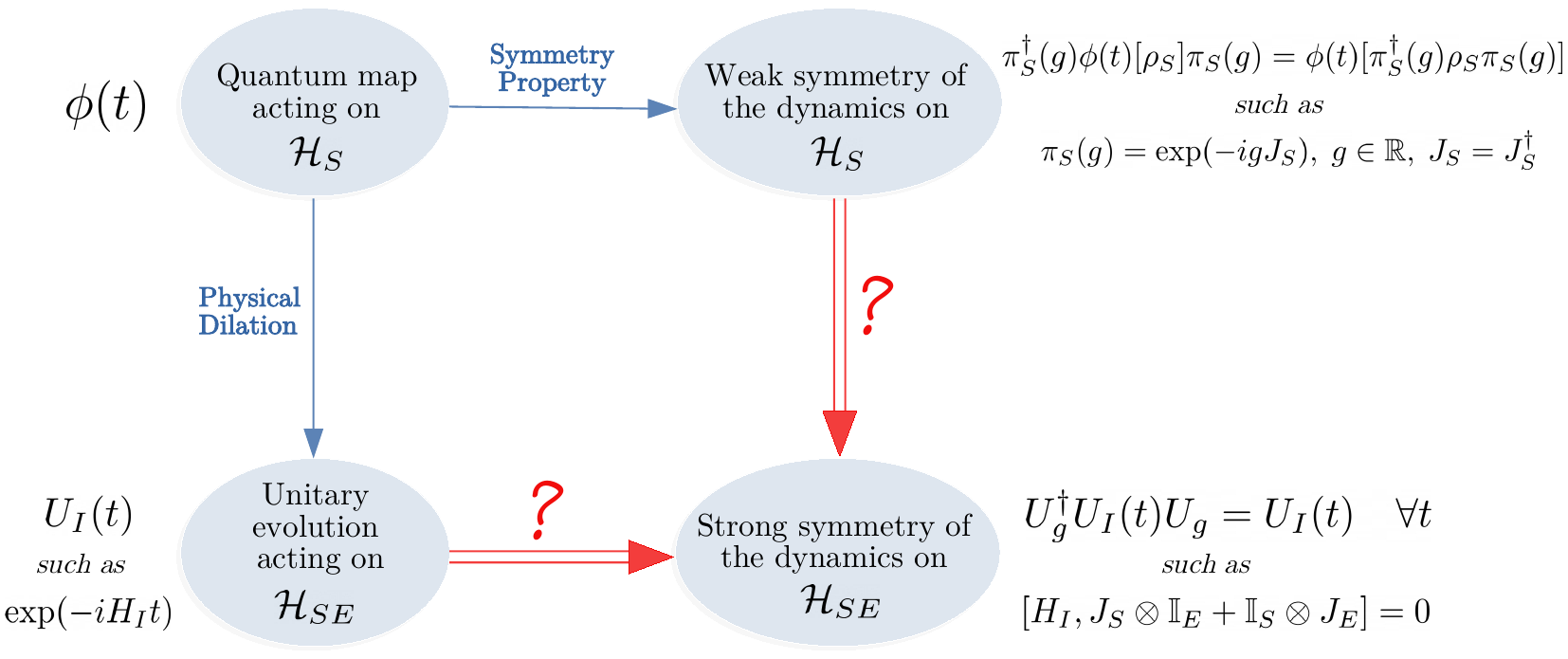}
    \caption{Conceptual framework of our research question (the notation used in the figure is introduced in Sec.~\ref{sec:background}). The map $\phi(t)$ acts on the system's Hilbert space $\mathcal{H}_S$ and is covariant with respect to a representation $\pi_S(g)$ of a symmetry group $G$. For example, if $G = (\mathbb{R},+)$, the symmetry representation is $\pi_S(g) = \exp(-i g J_S)$, where $J_S$ is a Hermitian operator of the system. A generic physical dilation of $\phi(t)$ can be constructed by introducing the environment's Hilbert space $\mathcal{H}_E$, the environment's initial state $\rho_E$, and the unitary evolution operator $U_I(t)$ acting on the joint space $\mathcal{H}_{SE} = \mathcal{H}_S \otimes \mathcal{H}_E$. Then, the map is expressed as $\phi(t) = \Tr_E[U_I(t)\rho_S \otimes \rho_E U_I^\dagger(t)]$. For instance, for dilations based on time-independent Hamiltonians, we have $U_I(t) = \exp(-i H_I t)$. The red arrows with question marks indicate the open problem we aim to address: given any dilation of $\phi(t)$ as described above, does a weak symmetry of $\phi(t)$ necessarily imply a strong symmetry of the dynamics on $\mathcal{H}_{SE}$? Specifically, for dilations involving time-independent Hamiltonians and $G = (\mathbb{R},+)$, can a conserved quantity always be found in the form $J_S \otimes \mathbb{I}_E + \mathbb{I}_S \otimes J_E$?}
    \label{fig:figure}
\end{figure*}

Symmetries play a central role in physics and are equally fundamental in the theory of quantum maps.  For instance, covariant quantum maps---those that respect certain symmetry principles---are useful to describe quantum information processes within the theory of quantum reference frames \cite{Bartlett2007,Gour2008,Boileau2008,Krumm2020} and quantum catalysis \cite{PhysRevLett.128.240501,lipka2024catalysis}. Moreover, symmetric maps can be used to describe the set of allowed operations within a quantum resource theory \cite{Chitambar2019}, and they impose very general constraints on how states, measurements, and channels transform under symmetric operations \cite{Marvian2008,Marvian2013,Ahmadi2013,Marvian2014,Marvian2014a,Marvian2014b,Marvian2016,Yamaguchi2023}. A prominent example of this is the characterization of quantum coherence as a resource \cite{Baumgratz2014}, which is rigorously explored within this framework \cite{Marvian2014a,Marvian2016a}, with crucial consequences for the  theory of quantum thermodynamics \cite{Lostaglio2015,Lostaglio2015a,Lostaglio2017,Lostaglio2019,Lostaglio2021}. Moreover, symmetries are essential for understanding the dynamics of open quantum systems \cite{Holevo1993,Holevo1996,Baumgartner2008,Buca2012,Albert2014}, particularly in the study of dissipative quantum phase transitions \cite{Minganti2018} and other phenomena that depend on the spectral properties of the generator of the dynamics \cite{Cattaneo2021b}.

Unlike in closed quantum systems, where Noether's theorem \cite{goldsteinCM} guarantees that every symmetry corresponds to a conserved quantity, symmetries in open systems and general quantum maps do not always imply conservation laws \cite{Buca2012,Marvian2014,Albert2014,Cirstoiu2020}. In these cases, we refer to \textit{weak symmetries}\footnote{Note that in this paper we will use the terms \textit{weakly symmetric} and \textit{covariant} (quantum maps) interchangeably.}, as opposed to \textit{strong symmetries}, which always lead to conserved quantities, as for the case of  unitary evolutions.

In this work, we explore the relationship between the theory of symmetries in quantum maps and open systems, including collision models, and the theory of physical dilations of quantum maps. To our knowledge, the only relevant result currently available in the literature is a theorem that guarantees the existence of a strongly symmetric dilation for every weakly symmetric quantum map \cite{Scutaru1979,Keyl1999,MarvianMashhad2012,Faist2021}. While this result is highly significant and has been widely utilized in recent years, it does not provide insights into the general constraints that a weak symmetry in a quantum map imposes on \textit{any} physical dilation of that map.

Understanding these constraints is crucial not only for theoretical reasons, such as for quantum resource theories and catalysis in quantum information, where dilations of covariant maps may play a significant role \cite{Chitambar2019,lipka2024catalysis}, but also for practical applications. As we briefly mentioned before, implementing physical dilations of quantum maps and open systems is increasingly important in quantum technologies, both for analog platforms and quantum computers. Consequently, it is essential to determine the rules governing the realization of dilations for covariant quantum maps. Specifically, given a weak symmetry in a map, must we always construct a strongly symmetric physical dilation? In other words, will there always be a conserved quantity in the dilated dynamics, or is it possible to break the symmetry at the level of the dilation? How does the structure of the dilation influence the result, particularly if we consider the very common case of time-dependent quantum maps? Fig.~\ref{fig:figure} offers a visual representation of our central research question.

A different way to frame this problem is as follows: consider a resource theory based on weakly symmetric quantum maps (see, e.g., \cite{MarvianMashhad2012,Marvian2016,Gour2008,Lostaglio2015,Chitambar2019}). Now, if we construct generic physical dilations for these quantum maps, can we always assert the presence of a strong symmetry at the level of the dilation? If so, does this symmetry map the resource theory at the level of quantum maps to a new, yet related, resource theory at the level of the dilations, driven by the emerging strong symmetry?

In this paper, we address the aforementioned questions by examining various types of time-dependent physical dilations of quantum dynamical maps, which are commonly encountered in quantum technologies and open systems. Specifically, we explore dilations driven by unitary operators generated by both time-independent and time-dependent Hamiltonians, the short-time limit of collision models, and dilations involving continuous, time-dependent unitary operators that do not form a one-parameter group. Additionally, we consider generic time-independent dilations for time-independent quantum maps, which include collision models with arbitrary timesteps. For each scenario, we analyze how the weak symmetry of the quantum map manifests at the level of its dilation. Furthermore, we provide a plethora of examples to support our theoretical results, which we believe are critical for gaining a deeper understanding of these concepts. 

The paper is structured as follows. In Sec.~\ref{sec:background} we introduce the formal definitions of quantum maps, dilations, and symmetries, along with the necessary mathematical tools for our analysis. Readers already familiar with these concepts may choose to skip this section. In Sec.~\ref{sec:defining}, we formally define our research questions and objectives, using the framework introduced in Sec.~\ref{sec:background}. Sec.~\ref{sec:results} presents our findings for the various types of physical dilations discussed earlier, with corresponding examples detailed in Sec.~\ref{sec:examples}, which include also a new result on the covariant dilation of the Landau-Streater channel \cite{Landau1993,Filippov2019}. Finally, in Sec.~\ref{sec:conclusions} we draw some concluding remarks.

\section{Background and formalism}
\label{sec:background}
In this section we introduce the necessary theoretical preliminaries for the formal description of symmetries and dilations of quantum maps. In this work, we restrict ourselves to finite-dimensional Hilbert spaces of the \textit{system} $\mathcal{H}_S$. We are interested in the evolution of the states of the system, which in general terms can be described by a \textit{quantum map} (alternatively called also \textit{quantum operation}, \textit{quantum process}, or \textit{quantum channel}) \cite{nielsenchuang,Holevo2012,Bengtsson2017}:
\begin{definition}[Quantum map] A quantum map $\phi$ is a linear, trace-preserving, and completely positive operator acting on the space of density matrices of the system $\mathcal{S}(\mathcal{H}_S)$:
\begin{equation}
\label{eqn:defQuantumMap}
    \phi: \mathcal{S}(\mathcal{H}_S)\rightarrow \mathcal{S}(\mathcal{H}_S).
\end{equation}
\end{definition}

If we are studying the system dynamics $\rho_S(t)$ continuously in time, its evolution at time $t$ is given by:
\begin{equation}
\label{eqn:actionQuantumMap}
    \rho_S(t)=\phi(t)[\rho_S(0)],
\end{equation}
where without losing generality we assume that the system dynamics starts at $t=0$. In this scenario, we usually refer to $\phi(t)$ as a \textit{quantum dynamical map}.

\begin{definition}[Dual map] 
Given any quantum map $\phi$, we can define the \textit{dual map}  $\phi^\dagger:\mathcal{B}(\mathcal{H}_S)\rightarrow\mathcal{B}(\mathcal{H}_S)$, where $\mathcal{B}(\mathcal{H}_S)$ is the space of bounded operators on $\mathcal{H}_S$, such that \cite{Holevo2012}:
\begin{equation}
\label{eqn:defDualMap}
    \Tr_S[\phi[\rho_S]A]=\Tr_S[\rho_S\phi^\dagger[A]],
\end{equation}
for all $\rho_S\in\mathcal{S}(\mathcal{H}_S)$ and for all $A\in\mathcal{B}(\mathcal{H}_S)$.
\end{definition}

It can be shown that $\phi^\dagger$ is linear, positive, and unital, meaning $\phi^\dagger[\mathbb{I}_S]=\mathbb{I}_S$. While the quantum dynamical map $\phi(t)$ drives the system dynamics in Schrödinger picture according to Eq.~\eqref{eqn:actionQuantumMap}, the associated dual map evolves the system observables in Heisenberg picture according to $  A(t) = \phi^\dagger(t)[A(0)]$.

The time evolution operator $U(t)$ satisfying the Schrödinger equation for an isolated quantum system can be described through a quantum map acting as $\phi(t)[\rho_S]=U(t)\rho_S U^\dagger(t)$. However, not all quantum maps are unitary. Quantum maps offer a more general framework for characterizing the dynamics of quantum systems, particularly in scenarios where the system is not isolated, such as when it interacts with an external environment or undergoes a non-selective measurement process. This naturally leads to the question: given a quantum map $\phi$, is it always possible to find a unitary evolution in an enlarged Hilbert space that reproduces the dynamics described by $\phi$ when restricted to the system's Hilbert space? This idea is central to the theory of dilations of quantum maps.

\subsection{Dilated quantum maps}
\label{sec:dilationMap}
\subsubsection{Stinespring's dilation theorem}
\label{sec:stinespring}
The cornerstone of the theory of dilated quantum maps is a formal result introduced by Stinespring in 1955 \cite{Stinespring1955}. Here we present Stinespring's theorem skipping the most formal details, which the interested readers can find in Appendix~\ref{sec:FormalStinespring}.

\begin{thm}[Stinespring's theorem  \cite{Paulsen2003}]
Given any dual map $\phi^\dagger$ acting on $\mathcal{B}(\mathcal{H}_S)$, we can always find a Hilbert space $\mathcal{H}_{SE}=\mathcal{H}_S\otimes\mathcal{H}_E$ and an isometry $V:\mathcal{H}_S\rightarrow\mathcal{H}_{SE}$ satisfying $V^\dagger V=\mathbb{I}_S$, such that, for all $A\in\mathcal{B}(\mathcal{H}_S)$,
\begin{equation}
    \label{eqn:stinespring}
    \phi^\dagger[A]=V^\dagger A\otimes\mathbb{I}_E V.
\end{equation}
\end{thm}

While the above result may sound quite abstract from a physical perspective, it puts forward the idea that any quantum map can always be described through a dilation $\mathcal{D}=(\mathcal{H}_{SE},V)$ in a larger Hilbert space $\mathcal{H}_S\otimes\mathcal{H}_E$ and an operation $V$ that preserves distances. From now on, we will refer to $\mathcal{H}_E$ as the Hilbert space of the \textit{environment}.

The consequences of Stinespring's theorem for the theory of quantum information cannot be overstated, as we will demonstrate in the next section. Before we proceed, however, we introduce another useful definition.

\begin{definition}[Minimal Stinespring dilation] A Stinespring dilation $\mathcal{D}=(\mathcal{H}_{SE},V)$ is \textit{minimal} \cite{Paulsen2003} if
\begin{equation}
\begin{split}
    &Span\{A\otimes\mathbb{I}_E V\ket{\psi_S},\, \forall A\in\mathcal{B}(\mathcal{H}_S),\,\forall \ket{\psi_S}\in\mathcal{H}_S\}\\
    &=\mathcal{H}_{SE}.
\end{split}
\end{equation}
In other words, minimal dilations ``cover'' the full enlarged Hilbert space $\mathcal{H}_{SE}$ if we compose the action of the isometry on $\mathcal{H}_S$ with any possible system operator. 
\end{definition}

If a Stinespring dilation is not minimal, then we can always find a minimal one by restricting the codomain of $V$ on a smaller $\mathcal{H}'_{SE}\subset\mathcal{H}_{SE}$, with $\mathcal{H}_{SE}'=Span\{A\otimes\mathbb{I}_E V\ket{\psi_S},\, \forall A\in\mathcal{B}(\mathcal{H}_S),\,\forall \ket{\psi_S}\in\mathcal{H}_S\}$, i.e., we ``remove'' the sectors of the enlarged Hilbert space that are not reached by the $Span$ defined above. 

It turns out that all minimal Stinespring dilations are unitarely equivalent, in the sense that if $\mathcal{D}_1=(\mathcal{H}_{SE}^{(1)},V_1)$ and $\mathcal{D}_2=(\mathcal{H}_{SE}^{(2)},V_2)$ are minimal, then clearly $dim(\mathcal{H}_{SE}^{(1)})=dim(\mathcal{H}_{SE}^{(2)})$, and there exists a unique unitary operator $U$ such that $V_2=U V_1$, $U A\otimes\mathbb{I}_{E^{(1)}} U^\dagger=A\otimes\mathbb{I}_{E^{(2)}}$ for all $A\in\mathcal{B}(\mathcal{H}_S)$ \cite{Paulsen2003}.

\subsubsection{Physical dilations of quantum maps}
\label{sec:physDilation}

In this work we are interested in the physical description of dilated quantum maps, so we want to express Stinespring's theorem in terms of physical entities that can be easily implemented in a lab, such as qubits and unitary evolutions. We can easily observe that, given a Stinespring dilation $\mathcal{D}=(\mathcal{H}_{SE},V)$, we can always find a unitary operator $U_I:\mathcal{H}_{SE}\rightarrow \mathcal{H}_{SE}$ and a fixed state $\ket{\psi_E}\in\mathcal{H}_E$ such that \cite{Caruso2006,Lindblad1976}:
\begin{equation}
\label{eqn:fromStinespringToPhysU}
    V\ket{\varphi_S}=U_I\ket{\varphi_S}\otimes\ket{\psi_E} \text{ for all }\ket{\varphi_S}\in\mathcal{H}_S.
\end{equation}
The readers can verify that $U_I$ is well-defined if $V$ is an isometry. Moreover, given a basis $\{\ket{\varphi_{S,j}}\}_j$ of $\mathcal{H}_S$,
\begin{equation}
    \begin{split}
        &\Tr_S[\phi^\dagger(A)\rho_S]=\Tr_S[V^\dagger A\otimes\mathbb{I}_E V \rho_S]\\
        &=\sum_j \bra{\varphi_{S,j}}\otimes\bra{\psi_E}U_I^\dagger A\otimes \mathbb{I}_E U_I \rho_S\ket{\varphi_{S,j}}\otimes\ket{\psi_E}\\
        &=\Tr_{SE}[U_I^\dagger A\otimes\mathbb{I}_EU_I \rho_S\otimes\ket{\psi_E}\!\bra{\psi_E}]=\Tr_S[A\phi[\rho_S]],
    \end{split}
\end{equation}
where the quantum map is eventually expressed as:
\begin{equation}
    \label{eqn:quantumMapPhysDil}
    \phi[\rho_S]=\Tr_E[U_I\rho_S\otimes\ket{\psi_E}\!\bra{\psi_E}U_I^\dagger].
\end{equation}
Eq.~\eqref{eqn:quantumMapPhysDil} is the standard representation of a dilation of a quantum operation in quantum information \cite{nielsenchuang}, which can be extended to mixed states of the environment $\rho_E\in\mathcal{S}(\mathcal{H}_E)$:
\begin{equation}
    \label{eqn:quantumMapPhysDilMixedState}
    \phi[\rho_S]=\Tr_E[U_I\rho_S\otimes\rho_E U_I^\dagger].
\end{equation}
We refer to Eq.~\eqref{eqn:quantumMapPhysDilMixedState} as a \textit{physical dilation} (or \textit{physical representation} \cite{Caruso2006}) of the quantum map $\phi$, characterized by $(\mathcal{H}_{SE},U_I,\rho_E)$. While any Stinespring dilation generates a physical dilation based on a pure state of the environment through Eq.~\eqref{eqn:fromStinespringToPhysU}, a physical dilation with a mixed state of the environment can be transformed into a Stinespring dilation by purifying $\rho_E$ through an additional ancillary system with Hilbert space $\mathcal{H}_C$ \cite{Caruso2006,nielsenchuang}.

\subsubsection{Quantum collision models}
\label{sec:collisionMod}

The quantum collision models \cite{Campbell2021a,Ciccarello2021,Cattaneo2022d,Cusumano2022} are a specific class of physical dilations that is of particular interest for the simulation of open quantum systems, i.e., quantum systems that interact with an external environment \cite{breuer2002theory,rivas2012open}. A single step of a collision model, lasting for a time $\Delta t$, is typically represented by the map $\phi_{\Delta t}$ written as in Eq.~\eqref{eqn:quantumMapPhysDilMixedState}, with\footnote{Note that throughout the text we assume $\hbar=1$.} $U_I(\Delta t)=\exp(-i H_I \Delta t)$. $H_I$ is some suitable interaction Hamiltonian that entangles the system with a single particle of the environment, which is usually termed \textit{ancilla}, prepared in $\rho_E$. After a single timestep, the ancilla is discarded and never interacts with the system again. In the following timestep, the system interacts with a new ancilla, identical to the previous one, which is again prepared in $\rho_E$. In this way, the discrete evolution of the state of the system can be obtained by repeated applications of the map $\phi_{\Delta t}$:
\begin{equation}
    \label{eqn:evCollMod}
    \rho_S(n\Delta t)=(\phi_{\Delta t})^n[\rho_S(0)].
\end{equation}
The discrete evolution of collision models makes them particularly suitable for the digital quantum simulation of open quantum systems \cite{Cattaneo2022d,Burger2022,Erbanni2023,Cattaneo2023}, which is known to be efficient \cite{Cattaneo2021}. Furthermore, when these models are combined with spatial evolution on a 1D lattice, they can serve as paradigmatic examples of quantum cellular automata \cite{Boneberg2023,Gillman2023}.

An especially interesting case arises in the limit of infinitesimal timestep, as the time derivative of the state of the system can then be written as:
\begin{equation}
\label{eqn:DerivativeCollision}
    \frac{d}{dt}\rho_S(t)= \lim_{\Delta t\rightarrow0^+}\frac{\phi_{\Delta t}[\rho_S(0)]-\rho_S(0)}{\Delta t}.
\end{equation}
Then, we expand the evolution operator up to the second order in $\Delta t$:
\begin{equation}
    \label{eqn:expansionCollModSecondOrder}
    U_I(\Delta t)\approx \mathbb{I}_{SE}-i \Delta t H_I-\frac{\Delta t^2}{2} H_I^2.
\end{equation}
Next, we plug the above equation into the expression for $\phi_{\Delta t}$, neglecting all the orders beyond $\Delta t^2$. If we also apply the standard assumption $\Tr_E[[H_I,\rho_S\otimes\rho_E]]=0$ for all $\rho_S$, then it can be shown that the collision model can describe the generator $\mathcal{L}$ of a \textit{quantum dynamical semigroup} $\exp(\mathcal{L}t)$ \cite{breuer2002theory,rivas2012open,Alicki2007} (see Appendix~\ref{sec:quantumDynSemig} for more details) as $\phi_{\Delta t}\approx \mathcal{I}_S+\Delta t \mathcal{L}$, and the action of the semigroup at any time $t$ is recovered through \cite{Ciccarello2021,Cattaneo2022d}:
\begin{equation}
    \label{eqn:collModSemigroup}
    \exp(\mathcal{L}t)=\lim_{\Delta t\rightarrow 0^+}(\phi_{\Delta t})^n,\; n=t/\Delta t.
\end{equation}
We refer the readers to Appendix~\ref{sec:collLindblad} for further details on the derivation of Eq.~\eqref{eqn:collModSemigroup}. 

While $\mathcal{L}$ constructed as above contains only the dissipative part of the dynamics, an additional unitary Hamiltonian evolution can also be easily included in $\mathcal{L}$ \cite{Landi2014,Lorenzo2017,Ciccarello2021,Cattaneo2022d}. Furthermore, through a suitable composition of different collisions the map $\phi_{\Delta t}$ can efficiently simulate \textit{any} generator $\mathcal{L}$ of a general quantum dynamical semigroup, including collective dissipators \cite{Cattaneo2021}. 

\subsection{Symmetries in quantum maps and open systems}
\label{sec:symmetries}
Recognizing the symmetries in the dynamics of classical or quantum systems is essential for gaining a deep understanding of any physical model. This is true not only for closed Hamiltonian models, in which Noether's theorem guarantees that continuous symmetries correspond to conserved quantities of the dynamics \cite{goldsteinCM}, but also for open systems and, in the context of quantum information, quantum maps. The study of symmetries, also known as \textit{covariances} (as will be clear from \textbf{Definition 4}), in quantum maps dates back nearly as far as the earliest formal investigations into quantum maps themselves \cite{Davies1970,Davies1970a,Scutaru1979}. Today, this topic has gained crucial importance in the fields of open quantum systems, quantum information, and quantum resource theories \cite{Bartlett2007,Vacchini2009,MarvianMashhad2012,Chitambar2019}.

\subsubsection{Definition and physical consequences}
\begin{definition}[Symmetry of a quantum map] Given a quantum map $\phi$ on $\mathcal{S}(\mathcal{H}_S)$, a symmetry group $G$, and a unitary representation thereof $\pi_S:G\rightarrow\mathcal{B}(\mathcal{H}_S)$, we say that $\phi$ is \textit{symmetric} or \textit{covariant} with respect to $G$ if
\begin{equation}
\label{eqn:defCovariance}
\pi_S^\dagger(g)\phi[\rho_S]\pi_S(g)=\phi[\pi_S^\dagger(g)\rho_S\pi_S(g)]
\end{equation}
 for all $\rho_S\in\mathcal{S}(\mathcal{H}_S)$ and $g\in G$. Equivalently,
\begin{equation}
\label{eqn:defSymm}
\mathcal{U}_S^\dagger(g)\phi\mathcal{U}_S(g)=\phi \text{ for all }g\in G,
\end{equation}
where $\mathcal{U}_S(g)[\rho_S]=\pi_S^\dagger(g)\rho_S\pi_S(g)$. For the purposes of this work, the representation $\pi_S$ is finite-dimensional.
\end{definition}

If $G$ is a one-dimensional connected Lie group, i.e., either $(\R,+)$ or $U(1)$, then its representation can be written as a one-parameter unitary group:
\begin{equation}
\label{eqn:phaseCov}
    \pi_S(g)=\exp(-i g J_S),
\end{equation}
where with abuse of notation $g$ is a real number and $J$ is some Hermitian operator acting on $\mathcal{H}_S$. Equivalently, we can define the superoperator\footnote{A superoperator is any linear operator acting on $\mathcal{B}(\mathcal{H}_S)$. In this paper,  superoperators are denoted by calligraphic letters, and any quantum map $\phi$ can also be seen as a superoperator.} $\mathcal{U}_S(g)=\exp(-i g \mathcal{J}_S)$, with
\begin{equation}
    \mathcal{J}_S=[J_S,\cdot],
\end{equation}
where $\cdot$ is a placeholder for a generic density matrix of the system. If $\phi$ is symmetric under $G$, then $[\phi,\mathcal{J}_S]=0$. Furthermore, in the case of a quantum dynamical semigroup $\phi(t)=\exp(\mathcal{L}t)$, we can express the symmetry as 
\begin{equation}
\label{eqn:weakSymmetryLiouvillian}
    [\mathcal{L},\mathcal{J}_S]=0.
\end{equation} 

We are particularly interested in the representation given in Eq.~\eqref{eqn:phaseCov} because it embodies one of the most significant symmetries for quantum maps and open systems, namely \textit{time-translation symmetry}. This symmetry is defined by $J_S = H$ and $g = t$, where $H$ is the system's Hamiltonian. Time-translation symmetry is crucial for defining suitable thermal operations \cite{Lostaglio2015,Lostaglio2017,Lostaglio2018,Lostaglio2019,Mohammady2022} and has been proposed as a requirement for open system dynamics that obey the laws of thermodynamics \cite{Dann2021,Dann2021a,Dann2022}. Additionally, time-translation symmetry with a slightly modified Hamiltonian, $J_S = H_0$, is satisfied by certain master equations for structured open systems, which remain valid in regimes where standard master equations with symmetry given by $J_S = H$ may fail \cite{Dann2021,Cattaneo2020,Trushechkin2021a}.

A particularly significant case of time-translation symmetry for a single qubit is \textit{phase covariance}. For a qubit with the Hamiltonian $H = \frac{\omega}{2} \sigma_z$, phase covariance is defined as the symmetry with $J_S = \sigma_z$ and $g \in [0, 2\pi)$, resulting in a two-dimensional representation of $U(1)$. Phase covariance is satisfied by several paradigmatic quantum dynamical semigroups \cite{Albert2014,Haase2019}, with diverse applications in metrology \cite{Smirne2016,Haase2018,Liuzzo-Scorpo2018} and the study of non-Markovianity \cite{Teittinen2018,Filippov2020}.

Finally, we emphasize that the condition $[\phi, \mathcal{J}_S] = 0$ does not necessarily imply that $J_S$ is a conserved quantity of the dynamics. More generally, suppose that we are interested in a continuous time-dependent quantum evolution described by $\phi(t)$, which is symmetric with respect to $G$ for all $t$. Then, crucially, Noether's theorem does not hold anymore for general open systems whose dynamics is not unitary, and this symmetry does not immediately imply the existence of a conserved quantity \cite{Marvian2014,Albert2014}. This fact is at the basis of the distinction between \textit{weak} and \textit{strong} symmetries \cite{Buca2012,DeGroot2022}.
\begin{definition}[Weak and strong symmetries]
    Suppose that the system dynamics is driven by a map $\phi(t)$ covariant with respect to the action of a group $G$. We say that $\phi(t)$ is \textit{weakly symmetric}\footnote{Alternatively, this is referred to as a \textit{symmetry on the superoperator level} \cite{Albert2014}.} with respect to $G$ if the symmetry does not lead to the emergence of a conserved quantity of the dynamics. On the contrary, if the symmetry is associated with a conserved quantity, we say that $\phi(t)$ is \textit{strongly symmetric}. 
\end{definition}

A symmetry of a Hamiltonian-driven evolution in a closed quantum system, for instance, is a strong symmetry, while symmetries in open systems can be weak. The implications of weak symmetries, as well as the distinct conditions under which they emerge compared to strong symmetries, have been extensively analyzed, particularly in the context of quantum dynamical semigroups \cite{Holevo1993,Holevo1996,Baumgartner2008,Buca2012,Albert2014}. These studies have found applications across a wide range of fields \cite{Manzano2014,Albert2016,Manzano2018,VanCaspel2018,Buca2019,Tindall2020a,Styliaris2019,Nigro2020,Kawabata2023}. 

Weak symmetries not only impose constraints on the structure of the dynamics but can also be leveraged to reduce the complexity of simulating it. For instance, Eq.~\eqref{eqn:weakSymmetryLiouvillian} implies that $\mathcal{L}$ can be block-diagonalized using a basis of eigenvectors of $\mathcal{J}_S$, which is often known \cite{Buca2012,Albert2014,Cattaneo2020,Bellomo2017,Cattaneo2021b,Dorn2021,Vaaranta2022}.

\subsubsection{Covariant Stinespring dilations}
\label{sec:CovStinespring}
After having defined covariant quantum maps, we run into a very natural question: what is the structure of Stinespring dilations for covariant maps? Can we find a refined version of \textbf{Theorem 1} in the presence of (weak) symmetries? To the best of our knowledge, this issue was first addressed by Scutaru \cite{Scutaru1979}, and then by Keyl and Werner \cite{Keyl1999}. Their results are summarized in the following theorem. 

\begin{thm}[Covariant dilations \cite{Keyl1999}]
Given any dual map $\phi^\dagger$ covariant under the action of the representation $\pi_S$ of some group $G$ according to \textbf{Definition 4}, we can always find a Stinespring dilation $\mathcal{D}=(\mathcal{H}_{SE},V)$ and another unitary representation of $G$, $\pi_E:G\rightarrow\mathcal{H}_E$, such that
\begin{equation}
    \label{eqn:stinespringCov}
    V \pi_S(g) = \pi_S(g)\otimes\pi_E(g) V.
\end{equation}
\end{thm}
In particular, the proof of this theorem is based on the fact that any minimal dilation (\textbf{Definition 3}) must satisfy Eq.~\eqref{eqn:stinespringCov}. Indeed, if $\mathcal{D}=(\mathcal{H}_{SE},V)$, then $\mathcal{D}_g=(\mathcal{H}_{SE}^{(g)},V_g)$ is also a minimal dilation, where $V_g=V\pi_S(g)$ and $\mathcal{H}_{SE}^{(g)}$ is a change of reference frame in the Hilbert space $\mathcal{H}_{SE}$ under the action of $\pi_S(g)\otimes\mathbb{I}_E$. For all practical purposes, we can think of this change as the replacement of $A\otimes\mathbb{I}_E$ in Eq.~\eqref{eqn:stinespring} with $\pi_S(g)A\otimes\mathbb{I}_E\pi_S^\dagger(g)$. Then, any pair of minimal dilations are unitary equivalent, so there must exist a unique operator $U_g$ such that $V_g=U_g V$ and $U_gA\otimes\mathbb{I}_EU_g^\dagger=\pi_S(g)A\otimes\mathbb{I}_E\pi_S^\dagger(g)$ for all $g\in G$. Then, $U_g= \pi_S(g)\otimes\pi_E(g)$, and any minimal dilation satisfies Eq.~\eqref{eqn:stinespringCov} \cite{Keyl1999}. This result will be crucial for the analysis of symmetries in physical dilations.

Marvian et al. also studied the issue of covariant dilations from a more physical perspective \cite{MarvianMashhad2012,Marvian2016} (see also the related Refs.~\cite{Lostaglio2019,Faist2021,Luijk2023}). In particular, they showed that one can always find a covariant physical dilation with a fixed state of the environment that is invariant with respect to the group action. Their result deals therefore with guaranteeing the existence of such a dilation. Our goal, instead, is to find the most general constraints that a symmetry imposes on any possible dilation of a quantum map.

\section{Defining the problem}
\label{sec:defining}

Having introduced the necessary tools in Sec.~\ref{sec:background}, we are now ready to formally define the problem we aim to address. We may formulate our general question as: \textit{Given a quantum evolution $\phi(t)$, which acts on the density matrices of a finite-dimensional Hilbert space $\mathcal{H}_S$ and is symmetric with respect to the action of a group $G$, what are the most general constraints that this symmetry imposes on any possible physical dilation $(\mathcal{H}_{SE},U_I(t),\rho_E)$ of $\phi(t)$, where $\mathcal{H}_E$ is also finite-dimensional? Does a weak symmetry of the map $\phi(t)$ according to \textbf{Definition 5} necessarily imply a strong symmetry of the unitary evolution governed by $U_I(t)$? How does the structure of $U_I(t)$ influence the results, including the case of time-independent maps with a fixed $U_I$?}

We focus on physical dilations that vary with time but have a fixed Hilbert space for the environment and a constant initial state $\rho_E$. This choice is rooted in experimental practice: when implementing a dilation of an open system dynamics in the laboratory, the most straightforward approach involves fixing the environment's physical system (i.e., $\mathcal{H}_E$) and the initial state $\rho_E$, while smoothly varying the generator of the quantum evolution $U_I(t)$ over time. Indeed, changing $\rho_E$ at each time instant would require modifying the experimental setup for each different $t$. Moreover, the choice of a fixed environment state $\rho_E$ is also motivated theoretically, e.g., as demonstrated by quantum collision models introduced in Sec.~\ref{sec:collisionMod}. Recent implementations of these models on near-term quantum computers \cite{Burger2022,Erbanni2023,Cattaneo2023} illustrate that the setup we are considering is also strongly aligned with current experiments in quantum technologies, including quantum simulations of open systems.

The question we pose has significant implications for both technological applications and quantum information theory. Specifically, understanding the constraints that symmetry imposes on the structure of the dilations of a generic quantum map defines the limitations and requirements for the experimental setups necessary to realize these dilations. Moreover, if a quantum resource theory is characterized by a particular covariance property of the quantum map, then analyzing how this covariance manifests in the dilation can reveal whether a new resource theory emerges at the level of the dilation. Fig.~\ref{fig:figure} provides a schematic visual summary of these concepts and of our central research question.

For the sake of exploring different possibilities for physical dilations that one may implement in the lab, we will consider five different types of $U_I(t)$, including the time-independent case (without losing generality the dynamics starts at $t=0$):
\begin{description}
    \item[Time-independent Hamiltonian evolution] This is the scenario in which we are mostly interested in, together with the collision models. The time evolution is driven by a fixed Hamiltonian $H_I$ as:
    \begin{equation}
    \label{eqn:timeIndHam}
        U_I(t)= \exp(-i H_It ).
    \end{equation}
    \item[Short-time collision models] We consider the short-time limit of collision models discussed in Sec.~\ref{sec:collisionMod}. The time evolution is given by 
    \begin{equation}
    \label{eqn:shortTimeCollTimeEv}
        U_I(\Delta t)=\exp(-i H_I \Delta t) \text{ for }\Delta t\rightarrow 0^+,
    \end{equation} and the quantum map is constructed according to the expansion in Eq.~\eqref{eqn:expansionCollModSecondOrder} and subsequent discussion (see Appendix~\ref{sec:collLindblad} for further details). This corresponds to the short-time limit of Eq.~\eqref{eqn:timeIndHam}. Note that a generic evolution written as Eq.~\eqref{eqn:timeIndHam} with finite-dimensional $\mathcal{H}_E$ cannot simulate a quantum dynamical semigroup \cite{VomEnde2023}. Then, the collision approach is an effective protocol for simulating Markovian master equations using the formalism of time-independent Hamiltonian evolution and finite-dimensional environment.
    \item[Time-dependent Hamiltonian] In this case the Hamiltonian generating the dynamics can vary as a function of time, and the evolution is given by:
    \begin{equation}
    \label{eqn:evOpTimeDepHam}
        U_I(t)=\mathcal{T}\exp\left(-i \int_0^t ds\, H(s) \right),
    \end{equation}
    where $\mathcal{T}$ is the time-ordering symbol \cite{fetterWalecka}. 
    \item[Continuous non-Hamiltonian evolution] We will also consider the quite abstract scenario in which $U_I(t)$ cannot be written as a Hamiltonian-driven evolution, but  is still continuous and derivable as a function of time. In general terms, $U_I(t)$ is not a one-parameter group as a function of time.
    \item[Time-independent map] Finally, we will consider the case in which $\phi$ does not depend on time, making $U_I$ a generic unitary operator without additional specific properties. This scenario encompasses collision models with a generic time step $\Delta t$. 
\end{description}

\section{Results}
\label{sec:results}
In this section we present the major results of our work. While we show some simple proofs here, we also skip the non-trivial details of the derivations, which the interested readers can find in Appendix~\ref{sec:derivation}. 

We consider a time-dependent quantum map $\phi(t)$ acting on a finite-dimensional $\mathcal{H}_S$ that is covariant with respect to the group $G$ and its representation $\pi_S(g)$. The starting point of our discussion is the fact that any minimal Stinespring dilation $\mathcal{D}(t)=(\mathcal{H}_{SE},V(t))$ of $\phi(t)$ must satisfy Eq.~\eqref{eqn:stinespringCov}, for some representation $\pi_E(g)$ of the group $G$ acting on $\mathcal{H}_{E}$. Then, we can construct a physical dilation $(\mathcal{H}_{SE},U_I(t),\ket{\psi_E})$ associated with the Stinespring dilation $\mathcal{D}(t)$ according to the discussion in Sec.~\ref{sec:physDilation}, and Eq.~\eqref{eqn:stinespringCov} becomes:
\begin{equation}
\label{eqn:physDilationCovStarting}
    \pi_S(g)\otimes\pi_E(g,t)U_I(t)\ket{\varphi_S}\otimes\ket{\psi_E}=U_I(t)\pi_S(g) \ket{\varphi_S}\otimes\ket{\psi_E}
\end{equation}
for all $g\in G$ and for all $\ket{\varphi_S}\in\mathcal{H}_S$. Eq.~\eqref{eqn:physDilationCovStarting} is the first step of our analysis. As explained in Sec.~\ref{sec:physDilation}, we can also construct physical dilations with a mixed state of the environment and then connect them to a well-defined Stinespring dilation.

We will now study separately the consequences of this equation for the different types of evolution operators listed in Sec.~\ref{sec:defining}. We primarily focus on the case of time-independent Hamiltonians, as the results obtained can be readily extended or appropriately modified to address the other scenarios. 

\subsection{Time-independent Hamiltonian evolution}
\label{sec:timeIndHamResult}
\subsubsection{Minimal dilations, pure state of the environment}
We first consider minimal physical dilations\footnote{Here we assume that the dilation is minimal for any time $t$ apart from a zero-measure set in $\R$, i.e., apart from some isolated $t^*\in\R$ for which the dilation may be non-minimal, e.g., the trivial case $t^*=0$.} constructed with a time evolution operator driven by a time-independent Hamiltonian $H_I$, given by Eq.~\eqref{eqn:timeIndHam}, and a pure state of the environment $\ket{\psi_E}$. Our first result is a property of the covariance in Eq.~\eqref{eqn:physDilationCovStarting} for time-independent Hamiltonian dilations, which we refer to as the \textit{stationary covariance property}:
\begin{prop}[Stationary covariance property] The representation $\pi_E$ of $G$ acting on the space of the environment in Eq.~\eqref{eqn:physDilationCovStarting} does not depend on time:
\begin{equation}
\label{eqn:statCov}
    \pi_E(g,t)=\pi_E(g)\text{ for all }t.
\end{equation}
\end{prop}
\begin{proof}
    In Appendix~\ref{sec:proofStatCovProp}.
\end{proof}
This is a quite important result, as it guarantees that the covariant property at the level of the dilated Hilbert space is constant, so that we can study it independently of time. Eq.~\eqref{eqn:physDilationCovStarting} becomes:
\begin{equation}
    \label{eqn:covPropertyTimeIndHam}
        U_ge^{-i H_I t}\ket{\varphi_S}\otimes\ket{\psi_E}=e^{-i H_I t}\pi_S(g) \ket{\varphi_S}\otimes\ket{\psi_E},
\end{equation}
where we have defined 
\begin{equation}
\label{eqn:repCovTotalHilbert}
    U_g=\pi_S(g)\otimes\pi_E(g).
\end{equation}
Eq.~\eqref{eqn:covPropertyTimeIndHam} has an immediate crucial consequence:
\begin{prop}[Invariant state of the environment] The state of the environment $\ket{\psi_E}$ is invariant under the group action, i.e., it is an eigenstate of $\pi_E(g)$ with eigenvalue 1:
\begin{equation}
\label{eqn:InvStateEnv}
    \pi_E(g)\ket{\psi_E}=\ket{\psi_E} \quad \forall g\in G.
\end{equation}
\end{prop}
\begin{proof}
    The result is obtained by taking the limit $t\rightarrow 0^+$ in Eq.~\eqref{eqn:covPropertyTimeIndHam}.
\end{proof}
\textbf{Proposition 2} imposes significant constraints on the structure of the possible representations $\pi_E$, and, consequently, on the feasibility of constructing appropriate physical dilations of the map $\phi(t)$ based on time-independent Hamiltonians. Specifically, Eq.~\eqref{eqn:InvStateEnv} requires that the projector $\ket{\psi_E} \!\bra{\psi_E}$ spans a one-dimensional subspace where $G$ operates via the trivial representation \cite{cornwell}. As a consequence, the representation $\pi_E$ must exhibit the corresponding block structure on $\ket{\psi_E}$.

A key result follows:
\begin{prop}[Strong symmetry in the subspace spanned by the dynamics] 
\begin{equation}
\label{eqn:SymmHamTimeIndRes}
   \restr{U_g^\dagger H_I U_g}{\Sigma_\parallel}=\restr{H_I}{\Sigma_\parallel}\quad \forall g\in G,
\end{equation}
where $\Sigma_\parallel$ is a subspace of $\mathcal{H}_{SE}$ defined by:
\begin{equation}
    \label{eqn:definitionParalSubspace}
\Sigma_\parallel=Span\{U_I(t)\ket{\varphi_S}\otimes\ket{\psi_E},\;\forall\ket{\varphi_S}\in\mathcal{H}_S,\;\forall t\}.
\end{equation}
\end{prop}
\begin{proof}
   In Appendix~\ref{sec:proofStatCovProp}.
\end{proof}
For time-independent Hamiltonians, $\Sigma_\parallel$ has a very nice characterization in terms of \textit{Krylov subspaces}, whose relevance for quantum information theory is growing fast \cite{Nandy2024,Parker2019}.
\begin{definition}[Krylov subspace] The order$-r$ Krylov subspace generated by the square matrix $A$ and the vector $\ket{v}$ is:
\begin{equation}
\label{eqn:krylovDef}
    \mathcal{K}_r(A,\ket{v})=Span\{\ket{v},A\ket{v},A^2\ket{v},\ldots,A^{r-1}\ket{v}\}.
\end{equation}
\end{definition}

Then, Eq.~\eqref{eqn:definitionParalSubspace} for $U_I(t)=\exp(-i H_I t)$ can be written as the space $\mathcal{K}_\parallel$, defined as:
\begin{equation}
\label{eqn:defKrylovParallel}
   \Sigma_\parallel= \mathcal{K}_\parallel:=\sum_{j=1}^n{\mathcal{K}_{r_{0,j}}\left(H_I,\ket{\varphi_{S,j}}\otimes\ket{\psi_E}\right)}.
\end{equation}
 $\{\ket{\varphi_{S,j}}\}_{j=1}^n$ is any basis of $\mathcal{H_S}$, while $\mathcal{K}_{r_{0,j}}$ is the Krylov subspace of maximal dimension for the state $\ket{\varphi_{S,j}}$. This subspace is characterized by the largest possible order $r_{0,j}$ such that Eq.~\eqref{eqn:krylovDef} is a span of linearly independent vectors\footnote{$r_{0,j}$ is also referred to as the \textit{Krylov dimension}}. Note that the summation in Eq.~\eqref{eqn:defKrylovParallel} denotes the sum (and not the direct sum) of vector spaces. Further details about this characterization and the Krylov subspaces can be found in Appendix~\ref{sec:krylov}.

We now consider the physical implications of \textbf{Proposition 3}. If $\phi(t)$ is covariant with respect to a \textit{weak symmetry}, as defined in  \textbf{Definition 5}, then \textit{any} physical dilation of $\phi(t)$ that is constructed as $(\mathcal{H}_{SE}, \exp(-i H_I t), \ket{\psi_E})$ must adhere to certain constraints. These constraints can be interpreted as the presence of a \textit{strong} symmetry\footnote{Indeed, the time-evolution Hamiltonian is symmetric according to Eq.~\eqref{eqn:SymmHamTimeIndRes}, leading to analogous consequences as symmetries in closed quantum systems.} within the system-environment dynamics, which holds only on a \textit{restricted subspace of} $\mathcal{H}_{SE}$. This subspace, intuitively, is spanned by the evolution at any possible time $t$ from any initial state of the system, while starting from the fixed initial state of the environment $\ket{\psi_E}$. In other words, it is the subspace explored by the system-environment dynamics. 

The consequences of this result are striking if we consider symmetries characterized by a continuous one-parameter group, as expressed in Eq.~\eqref{eqn:phaseCov}.
\begin{prop}[Conserved quantities in the subspace spanned by the dynamics] Suppose the quantum map $\phi(t)$ is symmetric with respect to $\pi_S(g)=\exp(-i g J_S)$ defined in Eq.~\eqref{eqn:phaseCov}. Then, the representation of the group on the space of the environment in Eq.~\eqref{eqn:repCovTotalHilbert} is characterized by $\pi_E(g)=\exp(-i g J_E)$, and 
\begin{equation}
\label{eqn:CommHamTimeInd}
   \restr{[H_I,J_S\otimes\mathbb{I}_E+\mathbb{I}_S\otimes J_E]}{\mathcal{K}_\parallel}=0,
\end{equation}
where $\mathcal{K}_\parallel$ is given by Eq.~\eqref{eqn:defKrylovParallel}.
\end{prop}
\textbf{Proposition 4} can be seen as an expression of Noether's theorem on the subspace of $\mathcal{H}_{SE}$ spanned by the dynamics, and it can be trivially extended to the generators of any continuous Lie group $G$ living in the corresponding Lie algebra, going beyond Abelian groups \cite{cornwell}. \textbf{Example 1} in Sec.~\ref{sec:examples} shows a simple scenario in which a strong symmetry emerges in the physical dilation of a phase-covariant map based on a time-independent Hamiltonian. 

The Hamiltonian $H_I$ and the representation of the symmetry group $U_g$ have a nice geometrical characterization in terms of the subspace $\mathcal{K}_\parallel$ in Eq.~\eqref{eqn:defKrylovParallel}:
\begin{prop}[Block-diagonal structure of the Hamiltonian and symmetry representation] 
If $P_\parallel$ is the projector on the subspace $\mathcal{K}_\parallel$ and, correspondingly, $P_\perp=\mathbb{I}-P_\parallel$ is the projector on the orthogonal complement $\mathcal{K}_\perp$, then the Hamiltonian and the representation $U_g$ on the full Hilbert space $\mathcal{H}_{SE}$ are block-diagonal in these subspaces:
\begin{equation}
    \label{eqn:blockDiagHiUg}
     \begin{split}
    H_I&=P_\parallel H_I P_\parallel\oplus P_\perp H_I P_\perp,\\
    U_g&=P_\parallel U_g P_\parallel\oplus P_\perp U_g P_\perp,\quad \forall g\in G.\\
 \end{split}
\end{equation}
\end{prop}
\begin{proof}
    In Appendix~\ref{sec:geometry}.
\end{proof}

The above proposition is showing that there is a freedom in the choice of the interaction Hamiltonian $H_I$ for the creation of a suitable physical dilation $(\mathcal{H}_{SE},\exp(-i H_I t),\ket{\psi_E})$. We can define a physical dilation of the same map $\phi(t)$ as $(\mathcal{H}_{SE},\exp(-i H_I' t),\ket{\psi_E})$. The new $H_I'$ must be identical to $H_I$ on $\mathcal{K}_\parallel$, while it can behave arbitrarily on $\mathcal{K}_\perp$. Then, we obtain the following natural result.
\begin{prop}[Construction of a symmetric Hamiltonian on the full system+environment space]
Take a physical dilation $(\mathcal{H}_{SE},\exp(-i H_I t),\ket{\psi_E})$ of a covariant map $\phi(t)$ that satisfies \textbf{Proposition 3}. Then, we can always find a Hamiltonian operator $H'_I$ which defines a physical dilation of $\phi(t)$ through $(\mathcal{H}_{SE},\exp(-i H_I' t),\ket{\psi_E})$ and is invariant with respect to the action of the global symmetry $U_g$ on the total Hilbert space $\mathcal{H}_{SE}$, i.e., $U_g H'_I U_g^\dagger=H'_I$ without restrictions to any subspace. Moreover, if both $H_I$ and $H_I'$ lead to suitable physical dilations of $\phi(t)$, then:
\begin{equation}
    \label{eqn:equalHamParallel}
    \restr{H_I}{\mathcal{K}_\parallel}=\restr{H'_I}{\mathcal{K}_\parallel}.
\end{equation}
\end{prop}
\begin{proof}
    In Appendix~\ref{sec:geometry}.
\end{proof}
\textbf{Proposition 6} is consistent with the results guaranteeing the existence of a covariant physical dilation that are already available in the literature \cite{MarvianMashhad2012,Marvian2016,Lostaglio2019,Faist2021,Luijk2023}.

\subsubsection{Non-minimal dilation, pure state of the environment}
We now suppose that the physical dilation we are interested in, $(\mathcal{H}_{SE},\exp(-i H_I t),\ket{\psi_E})$, leads to a non-minimal Stinespring dilation $(\mathcal{H}_{SE},V(t))$, with $V(t)=U_I(t)\ket{\psi_E}$. Then, we cannot straightforwardly apply the result in Eq.~\eqref{eqn:physDilationCovStarting}, which was obtained for minimal dilations. However, following the discussion in Sec.~\ref{sec:stinespring} we can construct a new minimal Stinespring dilation $(\mathcal{H}_{SE'},V'(t))$ on a Hilbert space $\mathcal{H}_{SE'}=\mathcal{H}_S\otimes\mathcal{H}_{E'}$ defined as \cite{Paulsen2003}:
\begin{equation}
    \label{eqn:minHilbertSpace}
    \begin{split}
    &\mathcal{H}_S\otimes\mathcal{H}_{E'}\\
    &=Span\{A\otimes\mathbb{I}_E V(t)\ket{\psi}_S,\, \forall A\in\mathcal{B}(\mathcal{H}_S),\,\forall \ket{\psi}\in\mathcal{H}_S\}.
    \end{split}
\end{equation}
We have assumed that the minimal dilated Hilbert space of the environment $\mathcal{H}_
{E'}$ does not change with time, apart from some isolated $t^*\in\R$ (e.g., for $t^*=0$). 

Next, we construct a new minimal physical dilation by restricting the action of $H_I$ to the action of a new $H_I':\mathcal{H}_{SE'}\rightarrow\mathcal{H}_{SE'}$. Note that a practical way to do so, following the $Span$ in Eq.~\eqref{eqn:minHilbertSpace} for any time $t$, is to consider the subspace defined in Eq.~\eqref{eqn:defKrylovParallel} and then extend it through the action of any possible operator $A\otimes\mathbb{I}_E$. Consequently, all the results presented in \textbf{Propositions 1} through \textbf{Proposition 6} remain valid for the restricted Hamiltonian $H_I'$ and the restricted environment Hilbert space $\mathcal{H}_{E'}$. 

In general terms, we will find a representation $U_g'=\pi_S(g)\otimes \pi_E'(g)$ of $G$ that acts only on $\mathcal{H}_{SE'}$. However, we can easily extend it to the full $\mathcal{H}_{SE}$ by defining $\pi_E(g)=\pi_E'(g)\oplus \mathbb{I}_{E'_\perp}$, where $\mathcal{H}_{E'_\perp}$ is the orthogonal complement of $\mathcal{H}_{E'}$ in $\mathcal{H}_E$, i.e., the Hilbert space that is not spanned by the minimal dilation defined according to Eq.~\eqref{eqn:minHilbertSpace}. The new $U_g=\pi_S(g)\otimes\pi_E(g)$ will then be a strong symmetry of the original $H_I$, with all the properties found in the previous section. 
 
\subsubsection{Mixed state of the environment}
\label{sec:mixedState}
We finally consider the case where the physical dilation $(\mathcal{H}_{SE},\exp(-i H_I t),\rho_E)$ is realized through a mixed state of the environment $\rho_E$, according to Eq.~\eqref{eqn:quantumMapPhysDilMixedState}. As explained in Sec.~\ref{sec:physDilation}, we can purify the state of the environment \cite{nielsenchuang} to obtain an equivalent non-minimal physical dilation with a pure state in an enlarged Hilbert space $\mathcal{H}_{SEC}$, say $\ket{\psi_{EC}}$. Then, if necessary, we can find the equivalent minimal dilation with pure state $\ket{\psi_{EC}}$ by means of the procedure explained in the previous subsection, and all the results from \textbf{Propositions 1} to \textbf{Proposition 6} follow accordingly.  \textbf{Example 2} in Sec.~\ref{sec:examples} illustrates this situation.

\subsection{Short-time collision models}
We now consider the case of a short-time collision model with evolution operator given by $U(\Delta t)$ in Eq.~\eqref{eqn:shortTimeCollTimeEv}. The short-time expansion is to be considered up to the second order in $\Delta t$ according to Eq.~\eqref{eqn:expansionCollModSecondOrder} (see also Appendix~\ref{sec:collLindblad} for further details). In this scenario, the symmetry of the map can be transformed into the symmetry of the Liouvillian $\mathcal{L}$ generating the quantum dynamical semigroup, according to Eq.~\eqref{eqn:collModSemigroup}.

Most of the results for the case of time-independent Hamiltonian evolution hold also for short-time collision models. Indeed, the stationary covariance property in \textbf{Proposition 1} is still valid, so Eq.~\eqref{eqn:covPropertyTimeIndHam} for a short $\Delta t$ is also true. Similarly, \textbf{Proposition 2} still holds as $U(\Delta t)\approx \mathbb{I}$ for $\Delta t \rightarrow 0^+$. Then, the Hamiltonian is still symmetric on the subspace spanned by dynamics $\Sigma_\parallel$ (Eqs.~\eqref{eqn:SymmHamTimeIndRes} and~\eqref{eqn:definitionParalSubspace}) according to \textbf{Proposition 3} and \textbf{Proposition 4}. The only difference lies in the structure of the subspace $\Sigma_\parallel$, which is not given by the sum of Krylov subspaces of maximal dimension in Eq.~\eqref{eqn:defKrylovParallel} anymore. Recalling the definition of a Krylov subspace in Eq.~\eqref{eqn:krylovDef}, the new subspace where the symmetry emerges is:
\begin{equation}
\label{eqn:krylovColl}
    \Sigma_\parallel= \mathcal{K}_\parallel^{(2)}:=\sum_{j=1}^n{\mathcal{K}_{r_{2,j}}\left(H_I,\ket{\varphi_{S,j}}\otimes\ket{\psi_E}\right)},
\end{equation}
for the usual basis $\{\ket{\varphi_{S,j}}\}_{j=1}^n$ of $\mathcal{H_S}$.
In other words, we consider only order-2 Krylov subspaces in the summation, as we also consider only the first two orders of $H_I$ in Eq.~\eqref{eqn:expansionCollModSecondOrder}. All the results of Sec.~\ref{sec:timeIndHamResult} hold with $\mathcal{K}_\parallel$ replaced by $\mathcal{K}_\parallel^{(2)}$. A nice example of this result can be found in \textbf{Example 3} in Sec.~\ref{sec:examples}.

\subsection{Time-dependent Hamiltonian}
Next, we consider the case of an evolution driven by a time-dependent Hamiltonian, with operator $U_I(t)$ given by Eq.~\eqref{eqn:evOpTimeDepHam}. It can be shown (see Appendix~\ref{sec:proofStatCovProp}) that \textbf{Proposition 1} and \textbf{Proposition 2} still hold, i.e., the stationary covariance property is valid, and the state of the environment must be invariant with respect to the action of the representation $\pi_E(g)$. So, we still have a symmetry on the subspace spanned by the dynamics. More precisely, \textbf{Proposition 3} still holds with a slightly refined formulation of Eq.~\eqref{eqn:SymmHamTimeIndRes}:
\begin{equation}
\label{eqn:symmGenericEv}
    \restr{U_g^\dagger H_I(t) U_g}{\Sigma_\parallel}=\restr{H_I(t)}{\Sigma_\parallel},\quad \forall g\in G,\; \forall t.
\end{equation}
The subspace  $\Sigma_\parallel$ spanned by the dynamics is given generically by Eq.~\eqref{eqn:definitionParalSubspace}. We do not have the nice characterization of $\Sigma_\parallel$ through Krylov subspaces anymore, so the remaining results in Sec.~\ref{sec:timeIndHamResult} cannot be trivially extended to the case of time-dependent Hamiltonians. 

\subsection{Continuous non-Hamiltonian evolution}
We now suppose that the Stinespring dilation is driven by an abstract non-Hamiltonian unitary operator $U_I(t)$, which is continuous and derivable for all times $t$ but does not form a one-parameter unitary group. We still assume that the Stinespring dilation with $V(t)=U_I(t)\ket{\psi_E}$ is minimal. If this is not the case, then we can find similar results by following the procedure for non-minimal dilations as in Sec.~\ref{sec:timeIndHamResult}.

We have not managed to prove the emergence of the stationary covariant property for such a generic scenario, but we have verified it in different examples if $U_I(t)$ is continuous and smooth. Therefore, let us assume that the representation $\pi_E$ satisfies the stationary covariance property. Then, we recover a symmetry in the subspace of the system+environment dynamics only if we can guarantee that $\pi_E(g)\ket{\psi_E}=\ket{\psi_E}$:
\begin{prop}[Symmetry for non-Hamiltonian evolutions] Suppose that $U_I(t)$ is continuous and derivable, and the representation $\pi_E$ is stationary. Next, suppose there exist a time $t^*$ at which 
\begin{equation}
\label{eqn:tstarU}
    U_I(t^*)\ket{\varphi_S}\otimes\ket{\psi_E}=\ket{\varphi_S'}\otimes\ket{\psi_E}
\end{equation}
  for all $\ket{\varphi_S}\in\mathcal{H}_S$. In other words, $U_I(t^*)$ acts as the identity on the Hilbert space of the environment. Then, 
  \begin{equation}
      \label{eqn:symmNonHamEv}    
     \restr{U_g^\dagger U_I(t) U_g}{\mathcal{H}_S\otimes Span\{\ket{\psi_E}\}}=\restr{U_I(t)}{\mathcal{H}_S\otimes Span\{\ket{\psi_E}\}}
  \end{equation}
  for all  $g\in G$ and for all $t$. If $U_g=\exp(-i g (J_S\otimes\mathbb{I}_E+\mathbb{I}_S\otimes J_E))$, then 
  \begin{equation}
      \restr{[J_S\otimes\mathbb{I}_E+\mathbb{I}_S\otimes J_E,U_I(t)]}{\mathcal{H}_S\otimes Span\{\ket{\psi_E}\}}=0,
  \end{equation} i.e., $J_S\otimes\mathbb{I}_E+\mathbb{I}_S\otimes J_E$ is a conserved quantity of the dynamics. 
\end{prop}
\begin{proof}
    It comes immediately from the stationary covariance property and the fact that there is a time $t^*$ for which the state of the environment is not modified by the application of $U_I(t)$ (see Eq.~\eqref{eqn:tstarU}), hence $\pi_E(g)\ket{\psi_E}=\ket{\psi_E}$, obtaining the assertion from Eq.~\eqref{eqn:physDilationCovStarting}.
\end{proof}

\textbf{Example 4} in Sec.~\ref{sec:examples} showcases a symmetric non-Hamiltonian dilation. In contrast, if the above conditions are not satisfied, then we can find physical dilations of covariant maps that do not display any symmetry of the closed system+environment dynamics, not even on a reduced subspace of $\mathcal{H}_{SE}$. The interested readers can find an example of such a dilation in \textbf{Example 5}.

\subsection{Time-independent map}
Finally, we consider a time-independent map $\phi$ covariant with respect to the action of $G$, and with physical dilation $(\mathcal{H}_{SE},U_I,\ket{\psi_E})$ (or analogously with a mixed state of the environment). In general terms, if no additional assumptions are made, then the covariant property of the map $\phi$ does not imply any constraints at the level of the dilation. In particular, $\ket{\psi_E}$ does not necessarily have to be  invariant under the action of $G$. If, instead, the representation $\pi_E(g)$ leaves $\ket{\psi_E}$ invariant, then as a trivial consequence of Eq.~\eqref{eqn:physDilationCovStarting} we obtain the same result as in Eq.~\eqref{eqn:symmNonHamEv} (removing the time dependence). Note that, once again, the symmetry does not necessarily appear on the full $\mathcal{H}_{SE}$, but only on the subspace $\mathcal{H}_S\otimes Span\{\ket{\psi_E}\}$. \textbf{Example 6} and \textbf{Example 7} in Sec.~\ref{sec:examples} illustrate these considerations. In particular, \textbf{Example 6} shows how to tackle scenarios with non-minimal dilations and an emerging strong symmetry in a subspace of $\mathcal{H}_{SE}$. In contrast, \textbf{Example 7} provides a Stinespring dilation of the Landau-Streater channel \cite{Landau1993,Filippov2019} that cannot give rise to any strong symmetry. 

While the weak symmetry of the map $\phi$ does not necessarily generate a strong symmetry of the unitary evolution $U_I$, it is always guaranteed that a symmetric dilation can be constructed by following the procedures in Refs.~\cite{MarvianMashhad2012,Faist2021}. In particular, we can always engineer a representation that acts trivially on some state of the environment $\ket{\psi_E}$, and build a suitable symmetric physical dilation with some invariant $U_I$ and $\ket{\psi_E}$. 

\section{Examples}
\label{sec:examples}

\subsection*{Example 1: Emergence of a strong symmetry with time-independent Hamiltonian evolution}
We define the quantum map $\phi(t)$ acting on a single qubit of the system:
\begin{equation}
\label{eqn:mapEx1}
\begin{split}
    &\phi(t)[\ket{1_S}\!\bra{1_S}]=\ket{1_S}\!\bra{1_S},\\&\phi(t)[\ket{0_S}\!\bra{0_S}]=\ket{0_S}\!\bra{0_S}\cos^2t+\ket{1_S}\!\bra{1_S}\sin^2t,    \\
    &\phi(t)[\ket{0_S}\!\bra{1_S}]=\ket{0_S}\!\bra{1_S}\cos t,\\&\phi(t)[\ket{1_S}\!\bra{0_S}]=\ket{1_S}\!\bra{0_S}\cos t.
\end{split}
\end{equation}
The readers can verify that this is a well-defined quantum map, and that it can be represented through the physical dilation
\begin{equation}
\label{eqn:Ex1PhysDil}
\phi(t)[\rho_S]=\Tr_E[U_I(t)\rho_S\otimes\ket{1_E}\!\bra{1_E} U_I^\dagger(t)].
\end{equation}
The environment is also made of a single qubit, and  $U_I(t)=\exp(-i H_I t)$, where $ H_I=(\sigma_S^+\sigma_E^-+H.c.)$. $\sigma_S^+$ is the raising operator of the system and analogously for $\sigma_E^+$. 

We introduce the one-parameter group whose elements are $\pi_S(g)=\exp(-i g N_S)$, where $N_S=\ket{1_S}\!\bra{1_S}$ and $g\in[0,2\pi)$. This is a representation of the group $U(1)$ made of the direct sum of two one-dimensional irreducible representations, one of them being trivial. Clearly, $\phi(t)$ is phase-covariant: 
\begin{equation}
\label{eqn:ex1PhaseCovariant}
    \pi_S(g)\phi(t)[\rho_S]\pi_S^\dagger(g)=\phi(t)[\pi_S(g)\rho_S\pi_S^\dagger(g)].
\end{equation} 
Note, however, that this phase covariance corresponds to a weak symmetry, as $N_S$ is not a conserved quantity of the dynamics. Let us investigate how this weak symmetry is reflected in the unitary dilation. 

With abuse of notation, we define the isometry of the Stinespring dilation as $V(t)=U_I(t) \ket{1_E}$. Introducing $V_g(t)=V(t)\pi_S(g)$, we construct another minimal Stinespring dilation\footnote{Note that the environment is a qubit, so any non-trivial dilation must be minimal.} according to the discussion in Sec.~\ref{sec:CovStinespring}. Then, we  can find a unitary $U_g$ such that $V_g(t)=U_g V(t)$ and $U_g=\pi_S(g)\otimes\pi_E(g)$. $\pi_E(g)$ does not depend on $t$ due to the stationary covariance property. Defining $N_E= \ket{1_E}\!\bra{1_E}$, we find:
\begin{equation}
\label{eqn:UgEx1}
U_g=\exp(-i g [N_S\otimes\mathbb{I}_E+\mathbb{I}_S\otimes (N_E-\mathbb{I}_E)]),
\end{equation}
corresponding to $\pi_E(g)=\exp(-i g  (N_E-\mathbb{I}_E))$.
Note that $\pi_E(g)\ket{1_E}=\ket{1_E}$, consistently with \textbf{Proposition 2}.

\textbf{Proposition 3} yields $\restr{U_g^\dagger H_I U_g}{\mathcal{K}_\parallel}=\restr{H_I}{\mathcal{K}_\parallel}$, where $\mathcal{K}_\parallel=Span\{\ket{0_S 1_E},\ket{1_S 1_E},H_I\ket{0_S 1_E}=\ket{1_S0_E}\}$, while $\mathcal{K}_\perp=Span\{\ket{0_S0_E}\}$. Note that $H_I$ and $U_g$ are block-diagonal in these subspaces, as predicted by \textbf{Proposition 5}. Moreover, we observe that $U_g^\dagger H_I U_g\ket{0_S 0_E}=H_I\ket{0_S 0_E}=0$. Therefore, the symmetry is valid all over the Hilbert space $\mathcal{H}_{SE}=\mathcal{H}_S\otimes\mathcal{H}_E$, and it corresponds to the immediate result that the total number of spin excitations is conserved according to \textbf{Proposition 4}, $[H_I,N_S+N_E]=0$. So, the weak phase-covariant symmetry of the map $\phi(t)$ induces a strong symmetry of the unitary dynamics in the dilated Hilbert space. 

Note that, while $U_g$ is a global symmetry on $\mathcal{H}_{SE}$, the relation $U_I(t)\pi_S(g)\otimes\mathbb{I}_E=U_g U_I(t)$ is valid only on the subspace $\mathcal{H}_S\otimes Span\{\ket{1_E}\}$. Moreover, if we considered the quantum map defined by the same dilation but with $H_I=\sigma_S^-\sigma_E^-+H.c.$, we would clearly obtain the same results but with conserved $N_S-N_E$.

Since $dim(\mathcal{K}_\perp)=1$, we cannot construct a Hamiltonian $H_I$ which gives rise to the same $\phi(t)$ and violates the global symmetry on $\mathcal{K}_\perp$ (for a scenario where such violation emerges we refer the readers to \textbf{Example 2}). However, note that we may employ a new Hamiltonian $H_I'$ such that, for instance, with abuse of notation $H_I'\mathcal{K}_\parallel=H_I\mathcal{K}_\parallel$, while $H_I'\ket{0_S0_E}=\ket{0_S0_E}$, and we would still obtain a faithful dilation of $\phi(t)$. This corresponds to the freedom in the definition of $H_I$ on the subspace $\mathcal{K}_\perp$, formalized by \textbf{Proposition 6}.  

\subsection*{Example 2: Time-independent Hamiltonian evolution and mixed state of the environment}
We consider a quantum map defined by a physical dilation that is analogous to Eq.~\eqref{eqn:Ex1PhysDil}, but with a mixed state of the environment:
\begin{equation}
\label{eqn:dilaSecondEx}
\phi(t)[\rho_S ]=\Tr_E[U_I(t)\rho_S\otimes\rho_E U_I^\dagger(t)],
\end{equation}
where both the system and the environment are individual qubits, $U_I(t)=\exp(-i  t (\sigma_S^+\sigma_E^-+H.c.))$, and $\rho_E=(1-c_1)\ket{1_E}\!\bra{1_E}+c_1\ket{0_E}\!\bra{0_E}$, $c_1\in(0,1)$. Once again, $\phi(t)$ is phase-covariant: $\pi_S(g)\phi(t)[\rho_S]\pi_S^\dagger(g)=\phi(t)[\pi_S(g)\rho_S\pi_S^\dagger(g)]$, where $\pi_S(g)=\exp(-i g N_S)$ is defined as in the previous example. 

Following the discussion in Sec.~\ref{sec:mixedState}, we purify the state $\rho_E$ by adding a qubit living in the Hilbert space $\mathcal{H}_C$, and the new state of the environment is: 
\begin{equation}
    \ket{\Psi_{EC}}=\sqrt{1-c_1}\ket{11_{EC}}+\sqrt{c1}\ket{00_{EC}}.
\end{equation} 
The dilation of the quantum map is then driven by $U_I(t)\otimes\mathbb{I}_C$, with interaction Hamiltonian $H_I=(\sigma_S^+\sigma_E^-+H.c.)\otimes\mathbb{I}_C$. The readers can verify that this physical dilation with pure state $\ket{\Psi_{EC}}$ corresponds to a minimal Stinespring dilation with $V(t)=U_I(t)\otimes\mathbb{I}_C\ket{\Psi_{EC}}$, for almost all times $t$. Following the standard procedure, we find $U_g=\pi_S(g)\otimes\pi_E(g)$ such that $V(t) \pi_S(g)=U_g V(t)$. We obtain:
\begin{equation}
    \pi_E(g)=\exp(i g \ket{0_E}\!\bra{0_E})\otimes\exp(-i g \ket{0_C}\!\bra{0_C}).
\end{equation}
Once again, observe that $\pi_E(g)\ket{\Psi_{EC}}=\ket{\Psi_{EC}}$. Note that $\mathcal{K}_\parallel=Span\{\ket{0_S0_E0_C},\ket{0_S1_E1_C},\ket{1_S0_E1_C},$ $\ket{1_S0_E0_C},\ket{0_S1_E0_C},\ket{1_S1_E1_C}\}$, while $\mathcal{K}_\perp=Span\{\ket{0_S0_E1_C},\ket{1_S1_E0_C}\}$. In accordance with \textbf{Proposition 3}, $\restr{U_g^\dagger H_I U_g}{\mathcal{K}_\parallel}=\restr{H_I}{\mathcal{K}_\parallel}\quad \forall g\in G$. 

Incidentally, we find $\restr{U_g^\dagger H_I U_g}{\mathcal{K}_\perp}=\restr{H_I}{\mathcal{K}_\perp}\quad \forall g\in G$ as well. This corresponds to the conservation of the operator $J=N_S+N_E-N_C$ over the whole Hilbert space $\mathcal{H}_S\otimes\mathcal{H}_E\otimes\mathcal{H}_C$.  Given that the Hilbert space $\mathcal{H}_C$ is not driven by the dynamics, $J'=N_S+N_E$ is conserved as well over $\mathcal{H}_S\otimes\mathcal{H}_E$, recovering the result of the first example.

The conservation of $J$ over $\mathcal{K}_\perp$ is not a requirement for a suitable physical dilation of $\phi(t)$. Indeed, let us construct a different physical dilation with the same Hilbert space $\mathcal{H}_{SE}$ and $\rho_E$ as in Eq.~\eqref{eqn:dilaSecondEx}, but with a different interaction Hamiltonian $H_I'$ in $U_I(t)$:
\begin{equation}
    \begin{split}
        H_I'=&(\sigma_S^+\sigma_E^-+H.c.)\otimes\mathbb{I}_C+\ket{0_S0_E1_C}\bra{1_S1_E0_C}\\&+\ket{1_S1_E0_C}\bra{0_S0_E1_C}.
    \end{split}
\end{equation} Note that $ \restr{H_I'}{\mathcal{K}_\parallel}=\restr{H_I}{\mathcal{K}_\parallel}$, where $H_I$ is the Hamiltonian of the previous example. Therefore, according to \textbf{Proposition 6}, $H_I'$ correctly generates the same quantum map $\phi(t)$. However, unlike in the previous case, the conservation of $J=N_S+N_E-N_C$ only holds in $\mathcal{K}_\parallel$, while $\restr{U_g^\dagger H_I' U_g}{\mathcal{K}_\perp}\neq\restr{H_I'}{\mathcal{K}_\perp}$. We thus observe an example of non-global conservation, where the constraint imposed by the phase-covariant map $\phi(t)$ generates a strong symmetry only on the subspace $\mathcal{K}_\parallel$ of the total Hilbert space.

\subsection*{Example 3: Differences between time-independent Hamiltonian evolution and short-time collision model}
Let us consider a quantum map acting on a \textit{ququart} living in $\mathcal{H}_S$, i.e., a $4$-level quantum system. We define it through its physical dilation with another ququart in $\mathcal{H}_E$:
\begin{equation}
    \phi(t)[\rho_S]=\Tr_E[U_I(t)\rho_S\otimes\ket{0_E}\!\bra{0_E}U_I^\dagger(t)],
\end{equation}
and $U_I(t)=\exp(-i H_I t)$, with the interaction Hamiltonian
\begin{equation}
\label{eqn:ex3HI}
    H_I = a_S^\dagger a_E+a_S a_E^\dagger.
\end{equation}
$a_S$ is the system annihilation operator truncated at the fourth level, and analogously for $a_E$. It is immediate to see that the map is phase-covariant with respect to the action of $\pi_S(g)=\exp(-i g a_S^\dagger a_S)$. 

We will consider both the map $\phi(t)$ continuous in time, and the map defined through a short-time collision model, namely $\phi(\Delta t)$ for $\Delta t\ll 1$. Note that $\Tr_E[a_E\ket{0_E}\!\bra{0_E}]=\Tr_E[a_E^\dagger\ket{0_E}\!\bra{0_E}]=0$, so the derivation of the map $\phi(\Delta t)$ is consistent with the discussion in Appendix~\ref{sec:collLindblad}. In particular, $\phi(\Delta t)[\rho_S]=\rho_S+\Delta t \mathcal{L}[\rho_S]$, and
\begin{equation}
\label{eqn:ex3Liou}
    \mathcal{L}[\rho_S]=\gamma_\downarrow \left(a_S\rho_S a_S^\dagger-\frac{1}{2}\{a_S^\dagger a_S,\rho_S\}\right), 
\end{equation}
with a suitable decay rate $\gamma_\downarrow$ (see Appendix~\ref{sec:collLindblad} for further details). 

Note that both the continuous-time map $\phi(t)$ and the map $\phi(\Delta t)$ for the collision model are defined through the same Hamiltonian $H_I$ in Eq.~\eqref{eqn:ex3HI}. It is immediate to realize that this Hamiltonian is symmetric with respect to $U_g=\pi_S(g)\otimes\pi_E(g)$, given
\begin{equation}
    \pi_E(g)=\exp(-i g a_E^\dagger a_E).
\end{equation}
This symmetry corresponds to the conservation of the total number of excitations in the two-ququart system: $[H_I,a_S^\dagger a_S+a_E^\dagger a_E]=0$, which is valid all over $\mathcal{H}_{SE}$. Clearly, $\pi_E(g)\ket{0_E}=\ket{0_E}$.

Next, let us analyze the subspace spanned by the dynamics driven by $H_I$ for the case of the continuous map. According to Eq.~\eqref{eqn:defKrylovParallel}, we find:
\begin{equation}
\begin{split}
    &\mathcal{K}_\parallel=Span\{\ket{0_S,0_E},\ket{1_S,0_E},\ket{2_S,0_E},\ket{3_S,0_E},\\
    &\ket{0_S,1_E},\ket{1_S,1_E},\ket{0_S,2_E},\ket{1_S,2_E},\ket{2_S,1_E},\ket{0_S,3_E}\}.    
\end{split}
\end{equation}
Note that this is, quite intuitively, the subspace of states having at most $3$ total excitations.
Conversely,
\begin{equation}
\begin{split}
    &\mathcal{K}_\perp=Span\{\ket{2_S,2_E},\ket{3_S,3_E},\ket{2_S,3_E},\ket{3_S,2_E}\}.    
\end{split}
\end{equation}

For the  map $\phi(\Delta t)$, instead, the subspace spanned by the dynamics is given by $\mathcal{K}_\parallel^{(2)}$ in Eq.~\eqref{eqn:krylovColl}, i.e., we do not have to include the states of the type $H^3\ket{\varphi_S}\otimes\ket{0_E}$ in the $Span$. The readers can verify that $\mathcal{K}_\parallel^{(2)}$ is formed by the $Span$ of all the vectors in $\mathcal{K}_\parallel$ apart from $\ket{0_S3_E}$. This means that we can still generate the same physical dilation of $\phi(\Delta t)$ in Eq.~\eqref{eqn:ex3Liou} by using a new Hamiltonian $H_I'=H_I+\ket{0_S3_E}\!\bra{2_S3_E}+\ket{2_S3_E}\!\bra{0_S3_E}$, which clearly violates the conservation of the total number of excitations outside $\mathcal{K}_\parallel^{(2)}$. However, we could not use $H_I'$ for the dilation of $\phi(t)$, as it would not be block-diagonal in $\mathcal{K}_\parallel$.

\subsection*{Example 4: Non-Hamiltonian continuous evolution with a strong symmetry of the dilation}
Once again, we consider a quantum map acting on a single qubit of the system and defined through a dilation with a single qubit of the environment: 
\begin{equation}
\label{eqn:mapNoGlobal}
\phi(t)[\rho_S]=\Tr_E[U_I(t)\rho_S\otimes\ket{1}_E\!\bra{1} U_I^\dagger(t)].
\end{equation}
In the basis $\{\ket{1_S1_E},\ket{1_S0_E},\ket{0_S1_E},\ket{0_S0_E}\}$,
\begin{equation}
U_I(t)=\begin{pmatrix}
\lambda_1(t)&-\lambda_0(t)&0&0\\
0&0&0&1\\
0&0&1&0\\
\lambda_0(t)&\lambda_1(t)&0&0\\
\end{pmatrix},
\end{equation}
with $\lambda_1(t)=\sqrt{\exp(-\gamma t)}$, $\lambda_0(t)=\sqrt{1-\exp(-\gamma t)}$, and $\gamma>0$. 
Note that the unitary operator is not a semigroup as a function of time: $U_I(t)U_I(s)\neq U_I(t+s)$. Moreover, $U_I(0)\ket{\varphi_S}\otimes\ket{1_E}=\ket{\varphi_S}\otimes\ket{1_E}$ for all $\ket{\varphi_S}\in\mathcal{H}_S$.

The action of the quantum map reads:
\begin{equation}
\begin{split}
    &\phi(t)[\ket{1_S}\!\bra{1_S}]=\ket{1_S}\!\bra{1_S}+(1-\exp({-\gamma t}))\ket{0_S}\!\bra{0_S},\\&\phi(t)[\ket{0_S}\!\bra{0_S}]=\ket{0_S}\!\bra{0_S},    \\
    &\phi(t)[\ket{0_S}\!\bra{1_S}]=\sqrt{\exp(-\gamma t)}\ket{0_S}\!\bra{1_S},\\&\phi(t)[\ket{1_S}\!\bra{0_S}]=\sqrt{\exp(-\gamma t)}\ket{1_S}\!\bra{0_S}.
\end{split}
\end{equation}
$\phi(t)$ makes the qubit decay toward the vacuum state. Clearly, the quantum map is once again phase-covariant: $\pi_S(g)\phi(t)[\rho_S]\pi_S^\dagger(g)=\phi(t)[\pi_S(g)\rho_S\pi_S^\dagger(g)]$, where $\pi_S(g)=\exp(-i g N_S)$ is defined as in the previous examples. 

The isometry $V(t)$ associated with the Stinespring dilation defined through Eq.~\eqref{eqn:mapNoGlobal} is:
\begin{equation}
    V(t)=\begin{pmatrix}
    \lambda_1(t)& 0\\
    0& 0\\
    0& 1\\
    \lambda_0(t)& 0\\
    \end{pmatrix}.
\end{equation}
If $V_g(t)=V(t)\pi_S(g)$, following the usual procedure we find $U_g$ such that $V_g(t)=U_g V(t)$:
\begin{equation}
    U_g=\exp(-i g N_S)\otimes\exp(-i g (\mathbb{I}_E-N_E)).
\end{equation}

$U_g=\pi_S(g)\otimes\pi_E(g)$ does not depend on time, so  it satisfies the stationary covariance property. Therefore, together with the fact that $U_I(0)$ acts as the identity on $\ket{\varphi_S}\otimes\ket{1_E}$, we immediately find $\pi_E(g)\ket{1}_E=\ket{1}_E$, and
\begin{equation}
    \restr{U_g^\dagger U_I(t) U_g}{\mathcal{H}_S\otimes Span\{\ket{1_E}\}}=\restr{U_I(t)}{\mathcal{H}_S\otimes Span\{\ket{1_E}\}},
\end{equation} 
for all $g\in G$, for all $t\geq 0$. $N_S-N_E$ is conserved when we start the dynamics from $\ket{\varphi_S}\otimes\ket{1_E}$. However, this is not true when we start it from $\ket{\varphi_S}\otimes\ket{0_E}$, and $U_g$ is a not a symmetry of the whole $\mathcal{H}_{SE}$, according to \textbf{Proposition 7}.

Note that we may replace the evolution $U_I(t)$ with one displaying global conservation of $N_S-N_E$ over the full $\mathcal{H}_{SE}$. Take, for instance,
\begin{equation}
U_I'(t)=\begin{pmatrix}
\lambda_1(t)&0&0&-\lambda_0(t)\\
0&1&0&0\\
0&0&1&0\\
\lambda_0(t)&0&0&\lambda_1(t)\\
\end{pmatrix}.
\end{equation}
The above evolution in Eq.~\eqref{eqn:mapNoGlobal} would still give rise to a physical dilation of the quantum map $\phi(t)$ with the same $V(t)$, and $N_S-N_E$ would be a conserved quantity of the dynamics in $\mathcal{H}_{SE}$. 

\subsection*{Example 5: Non-Hamiltonian continuous evolution without emerging strong symmetry}
Let us now consider a single-qubit quantum map defined through the Stinespring isometry:
\begin{equation}
     V(t)=\begin{pmatrix}
    1& 0\\
    0& -i \sin t\\
    0& \cos t\\
    0 & 0\\
    \end{pmatrix}.
\end{equation}
The above matrix is written in the basis $\{\ket{1_S},\ket{0_S}\}$ for the domain and $\{\ket{1_S1_E},\ket{1_S0_E},\ket{0_S1_E},\ket{0_S0_E}\}$ for the codomain. The readers can verify that the map defined by $\phi^\dagger(t)[A]=V^\dagger(t) A\otimes\mathbb{I}_E V(t)$ is that in Eq.~\eqref{eqn:mapEx1} of \textbf{Example 1}. We already know that $\phi(t)$ is phase-covariant, and $V(t)\pi_S(g)=U_gV(t)$, where $U_g=\pi_S(g)\otimes\pi_E(g)$ is given by Eq.~\eqref{eqn:UgEx1}. Note that the stationary covariance property is satisfied.

We will now construct a physical dilation of $\phi(t)$ based on a continuous non-Hamiltonian evolution $U_I(t)$, which does not possess any strong symmetry. We define:
\begin{equation}
    U_I(t)=\begin{pmatrix}
0&1&0&0\\
\cos t&0&0&-i\sin t\\
-i\sin t&0&0&\cos t\\
0&0&1&0\\
\end{pmatrix},
\end{equation}
and, with abuse of notation, $U_I(t)\ket{0_E}=V(t)$. Now, note that $\pi_E(g)\ket{0_E}=e^{ig}\ket{0_E}\neq\ket{0_E}$, since, according to Eq.~\eqref{eqn:UgEx1}, $\pi_E(g)=\exp{-i g  (N_E-\mathbb{I}_E)}$. Therefore, the readers can verify that $U_g$ is not a symmetry of $U_I(t)$ in any subspace of the global Hilbert space: $U_g^\dagger U_I(t) U_g \neq U_I(t)$. Hence, \textbf{Proposition 7} does not hold for this physical dilation. 

\subsection*{Example 6: Time-independent map and symmetry of the non-minimal dilation}
We consider a time-independent single-qubit quantum map $\phi$ that can be written in the basis $\{\ket{1_S}\!\bra{1_S},\ket{1_S}\!\bra{0_S},\ket{0_S}\!\bra{1_S},\ket{0_S}\!\bra{0_S}\}$ as:
\begin{equation}
    \phi=\begin{pmatrix}
    1&0&0&1/2\\
    0&1/\sqrt{2}&0&0\\
    0&0&1/\sqrt{2}&0\\
    0&0&0&1/2\\
    \end{pmatrix}.
\end{equation}
Once again, $\phi$ is clearly phase-covariant under the action of $\pi_S(g)=\exp(-i g N_S)$.

We can build a physical dilation of this map through a qutrit living in $\mathcal{H}_E$ and the unitary operator\footnote{The basis of the environment is ordered as $\ket{2_E}$, $\ket{1_E}$, $\ket{0_E}$.}
\begin{equation}
U_I=\begin{pmatrix}
0&0&0&1/\sqrt{2}&0&1/\sqrt{2}\\
-1/\sqrt{2}&1/\sqrt{2}&0&0&0&0\\
0&0&1&0&0&0\\
1/\sqrt{2}&1/\sqrt{2}&0&0&0&0\\
0&0&0&0&1&0\\
0&0&0&-1/\sqrt{2}&0&1/\sqrt{2}\\
\end{pmatrix}.
\end{equation}
The readers can verify that
\begin{equation}
\phi(\rho_S)=\Tr_E[U_I\rho_S\otimes\ket{0_E}\!\bra{0_E} U_I^\dagger].
\end{equation}

The isometry $V$ associated with the corresponding Stinespring dilation is:
\begin{equation}
    V=\begin{pmatrix}
    0&1/\sqrt{2}\\
    0&0\\
    1&0\\
    0&0\\
    0&0\\
    0&1/\sqrt{2}\\
    \end{pmatrix},
\end{equation}
showing that the subspace spanned by $\ket{1}_E$ is negligible, and the dilation is not minimal. It can be reduced to a minimal one by considering only the subspace $\mathcal{H}_{E_1}=Span\{\ket{2}_E,\ket{0}_E\}$, and the reduced operator $V'$ that, in the codomain basis $\{\ket{1_S2_E},\ket{1_S0_E},\ket{0_S2_E},\ket{0_S0_E}\}$, has the form
\begin{equation}
    V^{(restr)}=\begin{pmatrix}
    0&1/\sqrt{2}\\
    1&0\\
    0&0\\
    0&1/\sqrt{2}\\
    \end{pmatrix}.
\end{equation}
We now find $U_g^{(restr)}$ such that $V^{(restr)}\pi_S(g)=U_g^{(restr)}V^{(restr)}$. We obtain:
\begin{equation}
\label{eqn:ex6Ugprime}
    U_g^{(restr)}=\exp(- i g N_S)\otimes\exp(i g \ket{2_E}\!\bra{2_E}).
\end{equation}
Next, we define the unitary operator $\pi_E(g)=\exp(i g \ket{2_E}\!\bra{2_E})\oplus\ket{1_E}\!\bra{1_E}$ (i.e., we extend it to the full $\mathcal{H}_E$), and $U_g=\pi_S(g)\otimes\pi_E(g)$. 

We observe $U_g\ket{0}_E=\ket{0}_E$, therefore 
\begin{equation}
\label{eqn:ex6symm}
    \restr{U_g^\dagger U_I U_g}{\mathcal{H}_S\otimes\ket{0}_E}=\restr{U_I}{\mathcal{H}_S\otimes\ket{0}_E},
\end{equation} 
and $N_S\otimes\mathbb{I}_E-\mathbb{I}_S\otimes\ket{2}_E\!\bra{2}$ is conserved under the action of $U_I$ on the subspace $\mathcal{H}_S\otimes Span\{\ket{0_E}\}$. However, this conservation law does not hold when $U_I$ acts on any other state of the environment, and $U_g$ is a not a symmetry on the whole $\mathcal{H}_{SE}$.

Now, let us consider a slightly different physical dilation of $\phi$, based on the same qutrit space $\mathcal{H}_E$ and on the same initial state $\ket{0_E}$, but with different unitary operator:
\begin{equation}
U_I'=\begin{pmatrix}
0&0&0&1/\sqrt{2}&0&1/\sqrt{2}\\
0&0&1&0&0&0\\
-1/\sqrt{2}&1/\sqrt{2}&0&0&0&0\\
1/\sqrt{2}&1/\sqrt{2}&0&0&0&0\\
0&0&0&-1/\sqrt{2}&0&1/\sqrt{2}\\
0&0&0&0&1&0\\
\end{pmatrix}.
\end{equation}
The readers can verify that this is indeed a suitable non-minimal dilation of $\phi$. 
Then, the reduced space of the minimal dilation is $\mathcal{H}'_{E_1}=Span\{\ket{2}_E,\ket{1}_E\}$, and we find the same $U_g^{(restr)}$ as in Eq.~\eqref{eqn:ex6Ugprime}. By suitably extending $U_g^{(restr)}$ to the full $\mathcal{H}_E$, as we have done previously, we still obtain Eq.~\eqref{eqn:ex6symm}. Therefore, note that the strong symmetry of the dilation emerges in exactly the same way even if the initial state of the environment does not belong to the reduced subspace of the minimal dilation. 

\subsection*{Example 7: Covariant Stinespring dilation of the Landau-Streater channel}
The Landau-Streater channel is defined as \cite{Landau1993,Filippov2019}:
\begin{equation}
    \label{eqn:LandauStreater}
    \phi_{LS}[\rho]=\frac{1}{j(j+1)}\sum_{\alpha=1,2,3}J_\alpha^{(n)}\rho J_\alpha^{(n)},
\end{equation}
where $\rho$ is a density matrix on a Hilbert space of dimension $n=2j+1$, and $J_\alpha^{(n)}$ are $n$-dimensional generators of $\mathfrak{su}(2)$. This channel is covariant under the action of $SU(2)$ \cite{Filippov2019}:
\begin{equation}
    \label{eqn:landauCov}
    \pi_S(g)\phi_{LS}[\rho]\pi_S^\dagger(g)=\phi_{LS}[\pi_S(g)\rho\pi_S^\dagger(g)],
\end{equation}
where $g$ labels a generic group element, characterized by the rotation angle $\theta\in\R$ and the rotation unit vector $\mathbf{n}\in\R^3$:
\begin{equation}
\label{eqn:landauPig}
    \pi_S(g)=\exp(-i \theta \sum_{\alpha=1,2,3} n_\alpha J_\alpha^{(n)}),\qquad \sum_{\alpha=1,2,3} \abs{n_\alpha}^2=1.
\end{equation}

The Landau-Streater channel can be dilated into the Hilbert space\footnote{Note the inverted order in the definition of the dilated Hilbert space, which we introduce to keep the same notation as in Ref.~\cite{Filippov2019} and the elegant way of writing $V$ in Eq.~\eqref{eqn:LandauV}.} $\mathcal{H}_E\otimes\mathcal{H}_S$  through the isometry $V$, defined as:
\begin{equation}
\label{eqn:LandauV}
    V=\frac{1}{\sqrt{j(j+1)}}\begin{pmatrix}
    J_1^{(n)}\\
    
    J_2^{(n)}\\
    
    J_3^{(n)}\\
    \end{pmatrix}.
\end{equation}
Note that $dim(\mathcal{H}_E)=3$, and the dilation is always minimal. The overall Hilbert space has dimension $dim(\mathcal{H}_{SE})=3n=3(2j+1)$. We will now show that this dilation always satisfies a covariance property with respect to the adjoint representation of $SU(2)$ \cite{cornwell}. In particular, create the group representation $\pi_E(g)$ such that: 
\begin{equation}
    \label{eqn:landauCovT}
    \pi_E(g)=\exp(-i \theta \sum_{\alpha=1,2,3} n_\alpha T_\alpha^{(ad)}),\qquad \sum_{\alpha=1,2,3} \abs{n_\alpha}^2=1.
\end{equation}
$T_\alpha^{(ad)}$ are the $3$-dimensional generators of $\mathfrak{su}(2)$ giving rise to the adjoint representation:
\begin{equation}
\begin{split}
    &T_1^{(ad)}=\begin{pmatrix}
    0&0&0\\
    0&0&-i\\
    0&i&0
    \end{pmatrix},\quad
    T_2^{(ad)}=\begin{pmatrix}
    0&0&i\\
    0&0&0\\
    -i&0&0
    \end{pmatrix},\\
    &T_3^{(ad)}=\begin{pmatrix}
    0&-i&0\\
    i&0&0\\
    0&0&0
    \end{pmatrix}.
\end{split}
\end{equation}
These matrices satisfy $(T_\alpha^{(ad)})_{\beta\gamma}=-i\epsilon_{\alpha\gamma\beta}$, which are the structure constants of $\mathfrak{su}(2)$. Then, we have:
\begin{prop}[Covariance of the dilated Landau-Streater channel] Given the Landau-Streater channel in Eq.~\eqref{eqn:LandauStreater} and its Stinespring dilation in Eq.~\eqref{eqn:LandauV},
\begin{equation}
    \pi_E(g)\otimes\pi_S(g) V=V\pi_S(g).
\end{equation}
$\pi_S(g)$ and $\pi_E(g)$ are defined respectively in Eq.~\eqref{eqn:landauPig} and Eq.~\eqref{eqn:landauCovT}.
\end{prop}
\begin{proof}
For simplicity, from now on we will adopt Einstein's convention for summations over repeated indices. We observe:
\begin{equation}
V\pi_S^\dagger(g)=\frac{1}{\sqrt{j(j+1)}}\begin{pmatrix}
    J_1^{(n)}\pi_S^\dagger(g)\\
    
    J_2^{(n)}\pi_S^\dagger(g)\\
    
    J_3^{(n)}\pi_S^\dagger(g)\\
    \end{pmatrix}.    
\end{equation}
Therefore,
\begin{equation}
\begin{split}
   &\pi_E(g)\otimes\pi_S(g)V\pi_S^\dagger(g)
   \\
   &=\frac{1}{\sqrt{j(j+1)}}\begin{pmatrix}
    (\pi_E(g))_{1\alpha}\pi_S(g)J_\alpha^{(n)}\pi_S^\dagger(g)\\
    (\pi_E(g))_{2\alpha}\pi_S(g)J_\alpha^{(n)}\pi_S^\dagger(g)\\
    (\pi_E(g))_{3\alpha}\pi_S(g)J_\alpha^{(n)}\pi_S^\dagger(g)\\
    \end{pmatrix}.     
\end{split}
\end{equation}
Moreover, we know that $\pi_S(g)J_\alpha^{(n)}\pi_S^\dagger(g)=Q^\dagger_{\alpha\beta}J_\beta^{(n)}$, where $Q_{\alpha\beta}$ is the matrix describing the rotation in $\R^3$ with angle $\theta$ and unit vector $\mathbf{n}$ \cite{Filippov2019}. Additionally, the group elements generated by the adjoint representation $T_\alpha^{(ad)}$ transform like a representation of $SO(3)$, i.e., they generate the same rotation matrix: $(\pi_E(g))_{\alpha\beta}=Q_{\alpha\beta}$. Hence, we find $(\pi_E(g))_{\gamma\alpha}\pi_S(g)J_\alpha^{(n)}\pi_S^\dagger(g)=Q_{\gamma\alpha}Q^\dagger_{\alpha\beta}J_\beta=\delta_{\gamma\beta}J_\beta$, i.e., $\pi_E(g)\otimes\pi_S(g)V\pi_S^\dagger(g)=V$, which proves the assertion.
\end{proof}

The adjoint representation of $SU(2)$ is always irreducible as $\mathfrak{su}(2)$ is simple \cite{cornwell}. Therefore, it cannot be decomposed as the direct sum of the trivial representation and some other representations. Thus, we conclude that any physical dilation with $U_I$ of the Landau-Streater channel that is equivalent to the Stinespring dilation associated with $V=U_I\ket{\psi_E}$ and a $3-$dimensional environment, with $V$ given by Eq.~\eqref{eqn:LandauV}, cannot have a strong $SU(2)$ symmetry, as there is no state of the environment $\ket{\psi_E}$ that is invariant under the action of $\pi_E(g)$ in Eq.~\eqref{eqn:landauCovT} for any $g\in SU(2)$. So, we have found an example of a Stinespring dilation of a weakly symmetric quantum map that cannot be associated with any strongly symmetric physical dilation. 

\section{Conclusions}
\label{sec:conclusions}
In this paper, we have established a connection between between the theory of physical dilations of quantum maps and the theory of symmetries in open systems. Specifically, we have explored the constraints that the presence of a weak symmetry of the dynamics imposes on \textit{any} possible physical dilation of the corresponding quantum map $\phi(t)$. Our results show that for certain typical structures of time-dependent dilations, a weak symmetry necessarily implies a strong symmetry within the subspace spanned by the dynamics. However, this is not a universal result: in other types of dilations, such as the common case of time-independent dilations, a strong symmetry may neither be required nor guaranteed. Nevertheless, it remains true that a symmetric physical dilation for any weakly symmetric map can always be constructed.

In particular, if the physical dilation is driven by the standard unitary evolution operator arising from a time-independent Schrödinger equation $U_I(t)=\exp(-i H_I t)$, then a weak symmetry of the map $\phi(t)$  \textit{always} induces a strong symmetry of $H_I$ in the subspace spanned by the dynamics (\textbf{Proposition 3} and Eq.~\eqref{eqn:SymmHamTimeIndRes}), corresponding to a conserved quantity of the dynamics as expressed in \textbf{Proposition 4}. This subspace is, in general, not the full system-environment Hilbert space 
$\mathcal{H}_{SE}$, but it can be nicely characterized in terms of a sum of Krylov subspaces of maximal dimension according to Eq.~\eqref{eqn:defKrylovParallel}. \textbf{Example 1} and \textbf{Example 2} illustrate these general findings. 

The same results hold for the short-time limit of collision models, discussed in Sec.~\ref{sec:collisionMod} and Appendix~\ref{sec:collLindblad}, which is highly relevant for quantum simulations of the widely employed quantum Markovian semigroup. In this case, however, the structure of the subspace spanned by the dynamics is slightly modified, as given in Eq.~\eqref{eqn:krylovColl} and demonstrated in \textbf{Example 3}: we only have to consider order-2 Krylov subspaces in its construction. For both scenarios, there is always a clear way of extending the symmetry to the full Hilbert space $\mathcal{H}_{SE}$ through the characterization in terms of Krylov subspaces, as outlined in \textbf{Proposition 5} and \textbf{Proposition 6}. 

A strong symmetry is always guaranteed also for any physical dilation driven by a time-dependent Schrödinger equation. However, in this scenario we lose the characterization in terms of Krylov subspaces, and the subspace spanned by the dynamics is expressed in general terms according to Eqs.~\eqref{eqn:definitionParalSubspace} and~\eqref{eqn:symmGenericEv}.

In contrast, a strong symmetry of the dynamics in any subspace of $\mathcal{H}_{SE}$ is not a direct consequence of a weakly symmetric map if the latter does not depend on time, or if the physical dilation is driven by a generic $U_I(t)$ continuous in time that is not generated by a Schrödinger equation. A strong symmetry in a reduced subspace of $\mathcal{H}_{SE}$ can still emerge under certain conditions, namely if the initial state of the environment is invariant under the action of the symmetry group, as stated in \textbf{Proposition 7} and Eq.~\eqref{eqn:symmNonHamEv}. \textbf{Examples 4, 5, 6,} and~\textbf{7} illustrate these cases, highlighting both the emergence of strong symmetries and instances where no symmetry arises. We have also shown that some specific Stinespring dilations, such as the covariant dilation of the Landau-Streater channel in \textbf{Example 7}, cannot give rise to any strongly symmetric physical dilations due to the structure of their representation on the space of the environment.   

The findings summarized above provide a definitive answer to the research question illustrated in Fig.~\ref{fig:figure} for finite-dimensional systems. These results are particularly relevant for those aiming to implement physical dilations of quantum maps $\phi(t)$, e.g., for the sake of quantum simulations of open system dynamics. In the most common scenarios—dilations driven by a time-independent Hamiltonian or short-time collision models—we have demonstrated that the dilated dynamics is always subject to a hard constraint: it must display a strong symmetry with respect to the same symmetry group as the weak symmetry of $\phi(t)$. This imposes limitations on how physical dilations can be constructed; for example, it restricts the gate sequences available for implementing collision models that simulate covariant Markovian evolutions on a quantum computer.

A natural extension of this work would be to explore infinite-dimensional systems. We believe that, in practice, our results can often be applied straightforwardly even when the system is infinite-dimensional. Consider, for instance, the commonly encountered semigroup dynamics of a quantum harmonic oscillator, which is critical in setups like the dissipative Jaynes-Cummings model. The latter is widely used in quantum technologies, including circuit quantum electrodynamics and superconducting qubits \cite{Larson2021, Blais2020}. In practice, the bath inducing the semigroup evolution of the Jaynes-Cummings model is thermal at a finite temperature $T$. As a result, we can safely truncate the Hilbert space of the system at a certain level, depending on the temperature, with higher $T$ requiring more levels. Then, all the results of this paper can be immediately
applied to this truncated Hilbert space. However, more refined mathematical tools may be needed to generalize these findings to truly infinite-dimensional systems. 

In conclusion, this work establishes some fundamental principles for constructing physical dilations of covariant quantum maps, offering clear insights into when weak symmetries in the dynamics translate into strong symmetry constraints at the level of dilations, and when they do not. These findings hold significant theoretical value across various active fields, including quantum thermodynamics, quantum resource theory, and catalysis in quantum information. For instance, a direct consequence of our results is that, for many common dilations, a resource theory  imposing a weak symmetry property on the free maps can be mapped to a corresponding resource theory at the level of the dilation, now based on a strong symmetry. Beyond theoretical implications, our results have direct relevance for cutting-edge technological applications, such as quantum simulations on near-term quantum computers, where understanding the constraints and limitations due to symmetry is crucial for efficiently designing and optimizing suitable quantum algorithms.

\acknowledgements{We would like to thank Nicola Pranzini for reading the manuscript and, together with Otto Veltheim, for some useful tips on the proof of the stationary covariance property, as well as Sergey Filippov for suggesting the example on the Landau-Streater channel. We also acknowledge interesting discussions with Guillermo García-Pérez, Gian Luca Giorgi, Sabrina Maniscalco, Roberta Zambrini, and Paolo Zanardi. We acknowledge funding from the Research Council of Finland through the Centre of Excellence program grant 336810 and the Finnish Quantum Flagship project 358878 (UH), and from COQUSY project PID2022-140506NB-C21 funded by MCIN/AEI/10.13039/501100011033.}

\appendix

\section{A more rigorous formulation of Stinespring's theorem}
\label{sec:FormalStinespring}
While for all the practical purposes of this paper, with a finite-dimensional $\mathcal{H}_S$, we can employ the simplified expression for Stinespring's \textbf{Theorem 1} in Sec.~\ref{sec:stinespring}, we point out that its rigorous formulation is a bit more refined. Here, we follow the presentation given by Paulsen \cite{Paulsen2003}.   

First of all, the map $\phi$ does not need to be trace-preserving for Stinespring's theorem to hold, while complete positivity is sufficient. Moreover, the dual map formally acts as $\phi^\dagger:\mathcal{A}\rightarrow\mathcal{B}(\mathcal{H}_S)$, where $\mathcal{A}$ is a unital $C^*$-algebra acting on $\mathcal{H}_S$, i.e., the algebra of bounded operators on $\mathcal{H}_S$ closed under the adjoint operators, which contains the identity. For the sake of our work, we can identify $\mathcal{A}$ as $\mathcal{B}(\mathcal{H}_S)$. Moreover, the rigorous formulation of Stinespring's theorem makes use of a unital $*$-homomorphism $\eta:\mathcal{A}\rightarrow\mathcal{B}(\mathcal{H}_{SE})$, which is a transformation that conserves the algebra structure and the identity: $\eta(kA)=k\eta(A)$ for $k\in \C$, $\eta(AB)=\eta(A)\eta(B)$, $\eta(A+B)=\eta(A)+\eta(B)$, $\eta(A^\dagger)=\eta(A)^\dagger$, and $\eta(\mathbb{I})=\mathbb{I}$. $\mathcal{H}_{SE}$ is a generic Hilbert space that is larger than $\mathcal{H}_S$. With these ingredients, the Stinespring dilation in Eq.~\eqref{eqn:stinespring} is expressed as:
\begin{equation}
    \label{eqn:StinespringFormal}
    \phi^\dagger[A]=V^\dagger \eta[A] V,
\end{equation}
for the isometry $V$ and any $A\in\mathcal{A}$.

It can be shown that $\mathcal{H}_{SE}$ always has the tensor product structure $\mathcal{H}_S\otimes\mathcal{H}_E$, and that the requirement on $\eta$ as a unital $*$-homomorphism leads to $\eta(A)=A\otimes\mathbb{I}_E$ up to unitary transformations \cite{Keyl1999}, which is the formulation we have employed in the main text. Finally, note that if $\mathcal{D}_1=(\mathcal{H}_{SE}^{(1)},\eta_1,V_1)$ and $\mathcal{D}_2=(\mathcal{H}_{SE}^{(2)},\eta_2,V_2)$ are two minimal dilations  connected by the unitary $U$, then $V_2=UV_1$ and $U\eta_1U^\dagger=\eta_2$.

\section{Quantum dynamical semigroups}

\subsection{Definition}
\label{sec:quantumDynSemig}
A time-independent quantum dynamical semigroup is a family of quantum map $\phi(t)$ continuous in time that can be written as:
\begin{equation}
\label{eqn:quantumDynSemi}
    \phi(t)=\exp(\mathcal{L}t),
\end{equation}
where $\mathcal{L}$ is the generator (or \textit{Liouvillian}) of the semigroup. The master equation for the evolution of the state of the system simply reads
\begin{equation}
\label{eqn:masterEq}
    \frac{d}{dt}\rho_S(t)=\mathcal{L}[\rho_S(t)].
\end{equation}
This definition can be extended to time-dependent semigroups where $\mathcal{L}$ depends on time \cite{Chruscinski2022}. Quantum dynamical semigroups are usually referred to as ``Markovian''  due to their divisibility properties and the lack of traceable memory effects in the environment, although care must be taken when using the term ``Markovian'' with respect to its usual meaning in classical stochastic processes \cite{Wißmann2015,Chruscinski2022,Chruscinski2023,Canturk2024}.

A highly celebrated result is the geometrical characterization of the generator of any quantum dynamical semigroup in Eq.~\eqref{eqn:quantumDynSemi}, which can be written in the (non-diagonal) Gorini-Kossakowski-Sudarshan-Lindblad form (GKLS) \cite{breuer2002theory,Alicki2007,rivas2012open}:
\begin{equation}
\label{eqn:GKLSnondiag}
\begin{split}
    \mathcal{L}[\rho_S]=&-i[H_S,\rho_S]\\&+\sum_{j,k} \gamma_{jk}\left(A_j \rho_SA_k^\dagger-\frac{1}{2}\{A_k^\dagger A_j,\rho_S\}\right).
\end{split}
\end{equation}
$A_j$ are a family of operators that for convenience can be taken as orthonormal with respect to the Hilbert-Schmidt product, while $\gamma_{jk}$ is a positive matrix, and $H_S$ is a generic Hermitian operator acting on $\mathcal{H}_S$. 

Equivalently, we can write $\mathcal{L}$ in diagonal form by diagonalizing $\gamma_{jk}$ and finding a suitable family of \textit{Lindblad} operators $L_j$ \cite{Alicki2007}:
\begin{equation} 
    \label{eqn:GKLSdiag}
    \begin{split}
    \mathcal{L}[\rho_S]=&-i[H_S,\rho_S]\\
    &+\sum_{j} \Gamma_{j}\left(L_j \rho_SL_j^\dagger-\frac{1}{2}\{L_j ^\dagger L_j,\rho_S\}\right).
    \end{split}
\end{equation}
$\Gamma_j$ are the \textit{decay rates} of the dynamics.

\subsection{Derivation through a collision model}
\label{sec:collLindblad}
The importance of the semigroup dynamics in Eq.~\eqref{eqn:quantumDynSemi} cannot be overstated, as it describes the evolution of open systems weakly coupled to memoryless environments through an interaction that is captured by a physical dilation as in Eq.~\eqref{eqn:quantumMapPhysDilMixedState}. For instance, this formalism is widely employed in the field of quantum thermodynamics. For the environment to be memoryless, it needs to be infinite dimensional, so it can be shown that the physical dilations of quantum dynamical semigroups according to Eq.~\eqref{eqn:quantumMapPhysDilMixedState} cannot be constructed with a finite-dimensional environment \cite{VomEnde2023}. In contrast, we can implement finite-dimensional physical dilations of quantum dynamical semigroups by taking the short-time limit of collision models introduced in Sec.~\ref{sec:collisionMod} (this equivalence is exact, obviously, only for an infinitesimal timestep $\Delta t)$. Here, we briefly show how to implement the simplest version of this construction.

For simplicity, we neglect the unitary term driven by the Hamiltonian $H_S$ in Eq.\eqref{eqn:GKLSdiag}. We point out this term can be easily included in the collision model, see for instance Refs.~\cite{Landi2014,Lorenzo2017,Cattaneo2022d}. Instead, we focus on the dissipative part only, which is generated by repeated collisions written as a unitary dilations according to Eq.~\eqref{eqn:quantumMapPhysDilMixedState}, with $U_I(\Delta t)=\exp(-i g_I H_I \Delta t)$. $g_I$ is the magnitude of the collision energy. 

A single step of the collision model implements the action $\phi_{\Delta t}[\rho_S(t)]=\rho_S(t+\Delta t)$. Then, we can recover the time derivate of $\rho_S(t)$ continuously in time for $\Delta t\rightarrow 0^+$, according to Eq.~\eqref{eqn:DerivativeCollision}. The action of $\mathcal{L}$ in this limit can then be written as:
\begin{equation}
\label{eqn:liouvillianColl}
    \mathcal{L}[\rho_S]=\frac{\phi_{\Delta t}[\rho_S]-\rho_S}{\Delta t}.
\end{equation}
Note that, for $\Delta t\rightarrow 0^+$, $\mathcal{L}$ should not depend on $\Delta t$.

We now investigate the structure of  $\mathcal{L}$ defined in Eq.~\eqref{eqn:liouvillianColl}. First, we write the interaction Hamiltonian $H_I$ as:
\begin{equation}
\label{eqn:decomHI}
    H_I=\sum_\alpha A_\alpha\otimes B_\alpha,
\end{equation}
where $A_\alpha$ acts on $\mathcal{H}_S$ and $B_\alpha$ on $\mathcal{H}_E$. A requirement for the consistent derivation of $\mathcal{L}$ in the short-time limit is $\Tr_E[[H_I,\rho_S\otimes\rho_E]]=0$ for all $\rho_S$, where $\rho_E$ is the initial state of the environment ancilla in the dilation. To satisfy this condition, which is common for Markovian master equations \cite{breuer2002theory}, 
it is enough to require $\Tr_E[B_\alpha\rho_E]=0$ for all $B_\alpha$.

Next, we take the second-order expansion of $U_I(\Delta t)$ in $\Delta t$ given by Eq.~\eqref{eqn:expansionCollModSecondOrder} and plug it into the dilation Eq.~\eqref{eqn:quantumMapPhysDilMixedState}. We obtain:
\begin{equation}
\label{eqn:derCollExpan}
    \begin{split}
     &\phi_{\Delta t}[\rho_S]=\rho_S-\frac{\Delta t^2g_I^2}{2}\Tr_E[[H_I,[H_I,\rho_S\otimes\rho_E]]]\\
     &+\mathcal{O}(\Delta t^3)\\
     & \approx \rho_S+\gamma\Delta t\Tr_E[H_I\rho_S\otimes\rho_E H_I-\frac{1}{2}\{H_I^2,\rho_S\otimes\rho_E\}],
    \end{split}
\end{equation}
where we took the collision magnitude $g_I$ in such a way that $g_I^2\Delta t\rightarrow \gamma$ (i.e., fast but strong collision), which is a finite number independent of $\Delta t$. Finally, we plug Eq.~\eqref{eqn:derCollExpan} into Eq.~\eqref{eqn:liouvillianColl} and, using Eq.~\eqref{eqn:decomHI}, we find:
\begin{equation}
\begin{split}
    \mathcal{L}[\rho_S]=&\gamma \sum_{\alpha,\alpha'} \mu_{\alpha,\alpha'}\left(A_{\alpha'}\rho_SA_\alpha^\dagger-\frac{1}{2}\{A_{\alpha}^\dagger A_{\alpha'},\rho_S\}\right), 
\end{split}
\end{equation}
with $\mu_{\alpha,\alpha'}=\Tr_E[B_\alpha^\dagger B_{\alpha'}\rho_E]$. The generator $\mathcal{L}$ is in GKLS form according to Eq.~\eqref{eqn:GKLSnondiag}.

\section{Proofs of the results in Sec.~\ref{sec:results}}
\label{sec:derivation}

\subsection{Proof of the stationary covariance property and symmetry of the dynamics for Hamiltonian evolutions}
\label{sec:proofStatCovProp}
For a generic time-dependent Hamiltonian $H_I(t)$ and time evolution operator $U_I(t)$, we have:
\begin{equation}
    \frac{d}{dt}U_I(t)= -i H_I(t) U_I(t).
\end{equation}
Then, we derive Eq.~\eqref{eqn:physDilationCovStarting} with respect to time and obtain:
\begin{equation}
\begin{split}
    &\pi_S(g)\otimes\frac{d\pi_E(g,t)}{dt}U_I(t)\ket{\varphi_S}\otimes\ket{\psi_E}\\-&i\pi_S(g)\otimes\pi_E(g,t)H_I(t)U_I(t)\ket{\varphi_S}\otimes\ket{\psi_E}\\=&-iH_I(t)U_I(t)\pi_S(g) \ket{\varphi_S}\otimes\ket{\psi_E}\\
    =&-iH_I(t)\pi_S(g)\otimes\pi_E(g,t)U_I(t) \ket{\varphi_S}\otimes\ket{\psi_E},
\end{split}    
\end{equation}
where as usual this equation holds for all $t$, all $g\in G$, and all $\ket{\varphi_S}\in\mathcal{H}_S$.

The above equation is equivalent to
\begin{equation}
\label{eqn:commutStationCovOne}
    \restr{\pi_S(g)\otimes\frac{d\pi_E(g,t)}{dt}}{\Sigma_\parallel} =-i\restr{[H_I(t),\pi_S(g)\otimes\pi_E(g,t)]}{\Sigma_\parallel},  
\end{equation}
where $\Sigma_\parallel$ is the subspace spanned by the dynamics defined in Eq.~\eqref{eqn:definitionParalSubspace}. Eq.~\eqref{eqn:commutStationCovOne} resembles a von Neumann equation for the operator $\pi_S(g)\otimes\pi_E(g,t)$, where $\pi_S(g)$ does not depend on time. This means that $\pi_S(g)\otimes\pi_E(g,t)=U_I(t)\pi_S(g) \otimes \pi_E(g,0)U_I^\dagger(t)$ is a solution of Eq.~\eqref{eqn:commutStationCovOne}. Therefore, $H_I(t)$ leaves $\pi_S(g)$ invariant while evolving only $\pi_E(g,t)$, i.e., it does not create entanglement between the system and the environment, which is in contradiction with a non-trivial minimal physical dilation of the quantum map. So we conclude $\frac{d\pi_E(g)}{dt}=0$, proving the stationary covariant property.

Note that, after proving the stationary covariant property, from Eq.~\eqref{eqn:commutStationCovOne} we also immediately obtain the result in \textbf{Proposition 3} and Eq.~\eqref{eqn:SymmHamTimeIndRes}. Another way to show this for time-independent Hamiltonians consists in performing the series expansion of Eq.~\eqref{eqn:covPropertyTimeIndHam} as a function of time and then, after exploiting \textbf{Proposition 2}, equating each term of the same order $r$. These terms span the sum of Krylov subspaces in Eq.~\eqref{eqn:defKrylovParallel}, proving again the assertion.

\subsection{Characterization in terms of $\mathcal{K}_\parallel$ and $\mathcal{K}_\perp$}
\label{sec:geometry}
Once we have defined the subspace spanned by the dynamics according to Eq.~\eqref{eqn:defKrylovParallel} and Appendix~\ref{sec:krylov}, we can suitably construct the orthogonal component $\mathcal{K}_\perp$. We now suppose that there is a vector $\ket{b}\in\mathcal{K}_\perp$. This means that
\begin{equation}
    \bra{b}H_I^r\ket{\varphi_S}\otimes\ket{\psi_E}=0,\;\forall\ket{\varphi_S}\in\mathcal{H}_S,\;\forall r\in\N.
\end{equation}
Now, using the above equation and $H_I^\dagger=H_I$, it is immediate to see that $H_I$ is block diagonal in $\mathcal{K}_\parallel$ and $\mathcal{K}_\perp$. From this construction, it is also clear that any other block-diagonal Hamiltonian $H_I'$ with the same action as $H_I$ on $\mathcal{K}_\parallel$ will give rise to a suitable physical dilation of the map $\phi(t)$. In particular, we can simply construct a block-diagonal Hamiltonian $H_I'$ that acts as the identity on $\mathcal{K}_\perp$. Then, $H_I'$ is symmetric with respect to $U_g$ over the full Hilbert space $\mathcal{H}_{SE}$, proving \textbf{Proposition 6}.

Next, we evaluate the action of $U_g$ on a generic element of $\mathcal{K}_\parallel$:
\begin{equation}
    U_g H_I^r\ket{\varphi_S}\otimes\ket{\psi_E}= H_I^r\pi_S(g)\ket{\varphi_S}\otimes\ket{\psi_E}.
\end{equation}
We have used \textbf{Proposition 3} and \textbf{Proposition 2}. Clearly, $H_I^r\pi_S(g)\ket{\varphi_S}\otimes\ket{\psi_E}\in\mathcal{K}_\parallel$, so the action of $U_g$ on $\mathcal{K}_\parallel$ remains inside $\mathcal{K}_\parallel$. Using the fact that $U_g$ is unitary, we conclude 
\begin{equation}
    \bra{\varphi_S}\otimes\bra{\psi_E}H_I^r U_g \ket{b} =0,\;\forall\ket{\varphi_S}\in\mathcal{H}_S,\;\forall r\in\N,
\end{equation}
i.e., $U_g$ is block diagonal according to Eq.~\eqref{eqn:blockDiagHiUg}.

\section{Krylov subspaces}
\label{sec:krylov}
Krylov subspaces were introduced in 1931 as an iterative method to find the approximate solution of some classes of linear equations \cite{krylov1931}. Since then, Krylov subspaces have being widely used in linear algebra and computer sciences. Recently, they have received attention also in quantum information, as they can be employed to nicely describe the information spreading in the dynamics of many-body systems \cite{Nandy2024,Parker2019}. \textbf{Definition 6} and Eq.~\eqref{eqn:krylovDef} introduce the concept of a $r-$order Krylov subspace with matrix $A$ and vector $\ket{v}$. Here, we provide more details about Krylov subspaces and the characterization in Eq.~\eqref{eqn:defKrylovParallel}.

Given $A$ and $\ket{v}$ in the Hilbert space $\mathcal{H}$, we suppose that the maximal dimension of the Krylov subspace $\mathcal{K}_r(A,\ket{v})$ is $r=r_0$, i.e., the vectors $\ket{v},A\ket{v},\ldots,A^{r-1}\ket{v}$ are linearly independent until $r=r_0$. We call $r_0$ the \textit{Krylov dimension} of $\mathcal{K}(A,\ket{v})$. In other words, $\mathcal{K}_{r_0-1}(A,\ket{v})\subset\mathcal{K}_{r_0}(A,\ket{v})$, while $\mathcal{K}_{r\geq r_0}(A,\ket{v})\equiv\mathcal{K}_{r_0}(A,\ket{v})$.

In general terms, the vectors $\ket{v},A\ket{v},\ldots,A^{r_0-1}\ket{v}$ are linearly independent but not orthonormal with respect to a suitably defined inner product in $\mathcal{H}$. An orthonormal set of vectors spanning $\mathcal{K}_{r_0}(A,\ket{v})$ can be found by means of some routines such as the \textit{Lanczos algorithm} \cite{Nandy2024,lanczos1950iteration}. Note that this vector set is, in general, a basis for $\mathcal{K}_{r_0}(A,\ket{v})$ but not a basis for the full $\mathcal{H}$, since $\mathcal{K}_{r_0}(A,\ket{v})\subseteq\mathcal{H}$.

Note that, quite intuitively, if a vector $\ket{b}$ is in $\mathcal{K}_{r_0}(A,\ket{v})$, then $\mathcal{K}_{r_0^*}(A,\ket{b})\subseteq\mathcal{K}_{r_0}(A,\ket{v})$, where $r_0^*$ is the Krylov dimension of $\mathcal{K}(A,\ket{b})$. Indeed, $\ket{v},A\ket{v},\ldots,A^{r_0-1}\ket{v}$ is a basis of $\mathcal{K}_{r_0}(A,\ket{v})$. Then, $\ket{b}$ can be written as a linear combination of this basis and $A^r$ applied to $\ket{b}$ is still a linear combination of $\ket{v},A\ket{v},\ldots,A^{r_0-1}\ket{v}$ for any $r\in\N$.

We now consider the construction of the sum of Krylov subspaces in Eq.~\eqref{eqn:defKrylovParallel}. We start from the subspace $\Sigma_\parallel$ spanned by the dynamics in Eq.~\eqref{eqn:definitionParalSubspace}. Using the series expansion of $U_I(t)$,
\begin{equation}
\label{eqn:spanKrylovFirst}
\Sigma_\parallel=Span\{\sum_{r=0}^\infty \frac{(-it)^r}{r!}H_I^r\ket{\varphi_S}\otimes\ket{\psi_E},\;\forall\ket{\varphi_S}\in\mathcal{H}_S,\;\forall t\}.
\end{equation} 
Since we consider the vectors in the linear span at any time $t$, the above expression is equivalent to the span of the vectors for all orders of $H_I^r$:
\begin{equation}
\label{eqn:spanKrylovSecond}
\Sigma_\parallel=Span\{H_I^r\ket{\varphi_S}\otimes\ket{\psi_E},\;\forall\ket{\varphi_S}\in\mathcal{H}_S,\;\forall r\in\N\}.
\end{equation}

Fixing a single vector $\ket{\varphi_S}\otimes\ket{\varphi_E}$, Eq.~\eqref{eqn:spanKrylovSecond} is equivalent to $\mathcal{K}_{r_{0}}(H_I,\ket{\varphi_S}\otimes\ket{\varphi_E})$ for a suitable Krylov dimension $r_{0}$. Then, the span in Eq.~\eqref{eqn:spanKrylovSecond} for all vectors in $\mathcal{H}_S$ can be written as a sum\footnote{N.B.: we are speaking about a sum of vector spaces and not about a direct sum. In general terms, $\mathcal{K}_{r_{0,j}}(H_I,\ket{\varphi_{S,j}}\otimes\ket{\varphi_E})$ and $\mathcal{K}_{r_{0,j'}}(H_I,\ket{\varphi_{S,j'}}\otimes\ket{\varphi_E})$ for two different $\ket{\varphi_{S,j}}$ and $\ket{\varphi_{S,j'}}$ can be linearly dependent, or even the same vector space.} of Krylov subspaces starting from a generic basis of $\mathcal{H}_S$, which is $\mathcal{K}_\parallel$ in Eq.~\eqref{eqn:defKrylovParallel}. Indeed, fix a basis $\{\ket{\varphi_{S,j}}\}$ and build $\mathcal{K}_\parallel$ with this basis. Next, consider a vector $\ket{\phi_S}\in\mathcal{H}_S$ which is not an element of the basis. Then, clearly, $\ket{\phi_S}\otimes\ket{\psi_E}\in\mathcal{K}_\parallel$. Now, expanding $\ket{\phi_S}$ as a linear combination of $\{\ket{\varphi_{S,j}}\}$, we observe that $H_I^r\ket{\phi_S}\otimes\ket{\psi_E}$ must belong to the sum of Krylov subspaces $\mathcal{K}_\parallel$ for all $r\in\N$.

\bibliography{biblio}

\end{document}